\documentclass[journal]{IEEEtran}

\usepackage[T1]{fontenc}
\usepackage[utf8]{inputenc}
\usepackage{cite}
\usepackage{graphicx}
\usepackage{amsmath,amssymb,mathtools}
\usepackage{algorithm}
\usepackage{algpseudocode}
\usepackage{booktabs}
\usepackage{bm}

\newtheorem{proposition}{Proposition}

\newtheorem{remark}{Remark}
\newtheorem{assumption}{Assumption}

\begin{document}
	
\title{Pilot-Projected Nonlinear Self-Interference Estimation with Simultaneous Uplink Training}
	
	\author{Seyed Mohammad Azimi-Abarghouyi,~\IEEEmembership{Member,~IEEE}, and Thomas Eriksson,~\IEEEmembership{Senior Member,~IEEE}
		\thanks{The authors are with the Department of Electrical Engineering, Chalmers University of Technology, Gothenburg, Sweden (Emails: $\left\{\text{azimimo}, \text{thomase}\right\}$@chalmers.se). This work was supported by the Swedish Innovation Agency (Vinnova) Grant 2024-02531, MULTIRACS, through the Eureka CELTIC-NEXT Framework.}
	}
	
	\maketitle
	
\begin{abstract}
	This paper considers nonlinear digital self-interference (SI)
cancellation and uplink-channel estimation for single-antenna in-band
full-duplex radios when a desired uplink pilot is present during SI
training. In this setting, a rich nonlinear SI dictionary can absorb
part of the structured pilot response, while its large and correlated
set of basis functions can make reliable SI estimation difficult,
particularly with limited training. During SI-coefficient estimation,
we treat the unknown uplink channel as a nuisance parameter and
eliminate it to obtain an equivalent reduced SI estimation problem
on the orthogonal complement of the pilot convolution subspace. Since the projected nonlinear dictionary can remain ill conditioned, we develop
two complementary full-dictionary estimators: a pilot-projected ridge estimator
with Bayesian-information-criterion tuning (PP-RidgeBIC), which uses global
ridge shrinkage, and a pilot-projected group-sparse Bayesian estimator
(PP-GSBE), which learns branch-dependent Bayesian shrinkage from the data. These estimators offer two complementary operating points: PP-RidgeBIC favors computational efficiency, whereas PP-GSBE provides greater adaptivity to heterogeneous nonlinear SI components. Simulations show substantial gains in SI cancellation and uplink reception over non-projected and hard-selection alternatives; notably, PP-GSBE suppresses the residual SI below the noise floor with moderate training and consistently provides the strongest performance under nonlinear model mismatch.
\end{abstract}
	

\section{Introduction}

In-band full-duplex communication enables simultaneous transmission and
reception over the same time--frequency resource, and is therefore a promising
approach for improving spectral efficiency compared with conventional
half-duplex operation~\cite{Sabharwal2014,Zhang2016}. Its practical deployment,
however, is limited by self-interference (SI), where the signal transmitted by
the full-duplex node couples back into its own receiver. Since this coupled
signal can be much stronger than the desired received signal, full-duplex
receivers combine passive isolation, radio-frequency (RF)/analog cancellation,
and digital cancellation to suppress SI to a level compatible with reliable
detection~\cite{Choi2010,Jain2011,Bharadia2013,Duarte2012,Everett2014,
	Korpi2016,AhmedEltawil2015}. The last stage, digital cancellation, is an estimation
problem: the receiver uses its known transmitted samples to reconstruct the
residual SI at baseband and subtract it from the received signal.

Accurate digital SI cancellation is challenging because the residual SI is not
generally a linear filtered copy of the transmitted waveform. Hardware
impairments and frequency-selective coupling, including power-amplifier
nonlinearities, transmitter and receiver in-phase/quadrature (I/Q) imbalance,
oscillator phase noise,
quantization effects, and analog-canceller mismatch, can produce nonlinear,
memory-dependent, and conjugate-dependent SI components~\cite{Anttila2013,
	Ahmed2013,Korpi2014,Anttila2014,Korpi2017}. For this reason, nonlinear SI
cancellers are commonly built from behavioral baseband dictionaries generated
from the known transmit signal. These dictionaries may include delayed transmit
samples, polynomial and memory-polynomial terms, Volterra-type or
parallel-Hammerstein components, generalized memory-polynomial basis functions,
and widely-linear conjugate terms for image-related distortion~\cite{Morgan2006,
	Ding2004,Korpi2014}. While such models are expressive, they can also be large,
highly correlated, and ill conditioned, especially when high nonlinear orders,
long memory lengths, and both direct and conjugate branches are included.

Existing nonlinear SI cancellation methods mainly address this modeling and
estimation difficulty through adaptation, nonlinear representation, basis
conditioning, or model reduction. Augmented nonlinear least-mean-square (LMS)
filters account for nonlinear and conjugate-dependent SI components in adaptive
form~\cite{Li2018}. Neural-network-based cancellers learn nonlinear SI mappings
directly from data~\cite{Balatsoukas2018,Kurzo2018,Kristensen2019}.
Complexity-aware implementations reduce the cost of nonlinear digital
cancellation~\cite{Kristensen2020}. Orthonormalized polynomial bases improve
numerical conditioning by reducing correlations among nonlinear
regressors~\cite{Lee2025}. Another important line of work uses greedy model
selection to reduce a large nonlinear dictionary to a smaller active support
\cite{Shahghasi2025}. Such hard-selection methods can be coupled with
information criteria such as the Bayesian information criterion (BIC), where
BIC is used in its standard role: it compares discrete candidate supports by
balancing residual fitting error against the number of selected basis functions,
and therefore acts as a stopping or support-size selection rule~\cite{Schwarz1978}. These approaches improve important aspects
of nonlinear SI cancellation, but they usually treat the training observation as
a regression of SI in additive noise. They do not explicitly account for the
case where a desired uplink pilot is simultaneously present during SI training.

This paper studies this pilot-contaminated SI training regime. During the
training interval, the received baseband signal contains the residual SI driven
by the known local transmit waveform and the desired uplink pilot filtered by an
unknown uplink channel. The pilot sequence is known, but the channel coefficients
are not. Hence, the desired pilot response belongs to a known convolution
subspace with unknown coordinates. If the nonlinear SI dictionary is fitted
directly to the received mixture, a sufficiently rich dictionary can absorb part
of this structured pilot response into the estimated SI coefficients. This
causes two coupled errors: the learned SI model becomes biased, and the
subsequent uplink-channel estimate is degraded after imperfect SI subtraction.

Motivated by this observation, we develop a pilot-projected framework that uses a single simultaneous-training block for nonlinear SI estimation and uplink-channel estimation. Rather than treating the desired pilot contribution as unstructured interference during SI training, the framework explicitly accounts for its known structure and separates it from the nonlinear SI estimation problem. This enables reliable SI cancellation and subsequent channel recovery without requiring a dedicated SI-only training interval.

The main contributions of this paper are as follows.
\begin{itemize}
	
	\item We formulate nonlinear SI training in the presence of a simultaneous desired uplink pilot as a structured nuisance-parameter problem. By treating the unknown uplink-channel coefficients as unregularized nuisance parameters, we derive an equivalent reduced regression for SI estimation through projection onto the orthogonal complement of the known pilot convolution subspace.
	
	\item We employ a normalized widely-linear memory-polynomial dictionary whose
	direct and conjugate nonlinear branches define the coefficient groups, and
	develop two complementary pilot-projected full-dictionary estimators. The
	BIC-tuned pilot-projected ridge estimator (PP-RidgeBIC) uses a ridge-adapted
	Bayesian information criterion, based on the effective degrees of freedom, to
	tune a single global regularization parameter without hard basis selection,
	whereas the pilot-projected group-sparse Bayesian estimator (PP-GSBE) learns
	one variance hyperparameter per nonlinear branch to provide adaptive
	branch-wise shrinkage. The complexity analysis shows that PP-RidgeBIC provides
	the lower-complexity solution, while PP-GSBE requires additional iterative
	computation in exchange for greater adaptivity.
	
	\item We characterize the analytical consequences of pilot projection and residual SI estimation error and evaluate the proposed framework across the complete receiver chain, including training-phase SI estimation and cancellation, uplink-channel estimation, data-phase SI reconstruction, and data detection. Analytically, we quantify the fraction of each nonlinear SI basis waveform retained after pilot projection, thereby characterizing its post-projection identifiability, and derive how residual SI reconstruction error propagates through the pilot matched-filtering operation into the subsequent uplink-channel estimate. Numerically, we compare the proposed methods against non-projected regularized estimation, pilot-projected greedy model selection, and pilot-projected orthonormalized adaptive nonlinear cancellation using SI cancellation gain, residual interference-to-noise ratio (INR), channel normalized mean-square error (NMSE), bit-error rate (BER), and error-vector magnitude (EVM). The results show that PP-GSBE achieves the highest SI cancellation gain, suppresses the residual SI below the noise floor at moderate training lengths, and provides the lowest channel NMSE, BER, and EVM over the considered training range. It also provides the strongest residual-SI suppression under variations in the pre-digital residual SI level, power-amplifier nonlinearity, and I/Q image leakage. PP-RidgeBIC, meanwhile, follows PP-GSBE well across most operating conditions and performance metrics.
	
\end{itemize}

\subsubsection*{Notation}
	
	Scalars are denoted by lower-case letters, vectors by bold lower-case letters,
	and matrices by bold upper-case letters. The operators $(\cdot)^T$,
	$(\cdot)^H$, and $(\cdot)^*$ denote transpose, Hermitian transpose, and complex
	conjugation, respectively. The Euclidean norm is denoted by $\|\cdot\|_2$.
	For a matrix $\mathbf{A}$, $\mathrm{rank}(\mathbf{A})$ and
	$\mathrm{tr}(\mathbf{A})$ denote its rank and trace. The identity matrix is
	denoted by $\mathbf{I}$.
	
	A circularly symmetric complex Gaussian vector with mean $\bm{\mu}$ and
	covariance matrix $\mathbf{C}$ is denoted by
	\begin{equation}
		\mathbf{x}\sim\mathcal{CN}(\bm{\mu},\mathbf{C}).
	\end{equation}
	The local transmit signal generated by the full-duplex node is denoted by $x[n]$, and the
	uplink pilot transmitted by the remote user is denoted by $t[n]$. Since $x[n]$
	is generated locally, it is assumed to be known at the full-duplex receiver. Therefore,
	all nonlinear and conjugate basis functions generated from $x[n]$ are available
	for digital cancellation.

\section{System Model}
\label{sec:system_model}
	
	\subsection{Received Training Signal}
	
	Consider a single-antenna in-band full-duplex node that transmits and receives
	simultaneously over the same frequency band. During a training interval, the
	received baseband signal is
	\begin{equation}
		y[n]=y_u[n]+y_{\mathrm{SI}}[n]+w[n],
		\label{eq:received_signal}
	\end{equation}
	where $y_u[n]$ is the desired uplink signal, $y_{\mathrm{SI}}[n]$ is the
	residual SI after passive suppression and analog cancellation, and
	$w[n]\sim\mathcal{CN}(0,\sigma^2)$ is receiver noise.
	
	The desired uplink component is modeled as an $L_u$-tap channel:
	\begin{equation}
		y_u[n]
		=
		\sum_{\ell=0}^{L_u-1}h_u[\ell]\,t[n-\ell],
		\label{eq:uplink_signal}
	\end{equation}
	where $h_u[\ell]$ denotes the $\ell$th tap of the uplink channel impulse response.
	
	Stacking $N$ received training samples gives
	\begin{equation}
		\mathbf{y}
		=
		\mathbf{S}\mathbf{h}_u
		+
		\mathbf{y}_{\mathrm{SI}}
		+
		\mathbf{w},
		\label{eq:stacked_received}
	\end{equation}
where $\mathbf{y}\in\mathbb{C}^{N}$ is the received vector,
$\mathbf{h}_u\in\mathbb{C}^{L_u}$ is the uplink channel vector,
$\mathbf{S}\in\mathbb{C}^{N\times L_u}$ is the pilot convolution matrix,
and $\mathbf{y}_{\mathrm{SI}},\mathbf{w}\in\mathbb{C}^{N}$ are the stacked
residual-SI and receiver-noise vectors, respectively.
	
For $i=0,\ldots,N-1$, let $n_i=n_0+i$, where $n_0$ is the first
training-sample index. The $(i+1)$th row of $\mathbf{S}$ is
	\begin{equation}
		\mathbf{s}^{T}[n_i]
		=
		\begin{bmatrix}
			t[n_i] & t[n_i-1] & \cdots & t[n_i-L_u+1]
		\end{bmatrix},
	\end{equation}
	so that
	\begin{equation}
		\mathbf{S}
		=
		\begin{bmatrix}
			\mathbf{s}^{T}[n_0]\\
			\mathbf{s}^{T}[n_0+1]\\
			\vdots\\
			\mathbf{s}^{T}[n_0+N-1]
		\end{bmatrix}.
	\end{equation}
	
	\subsection{Widely-Linear Nonlinear SI Dictionary}
	\label{sec:wl_dictionary}
	
	The residual SI is driven by the known local transmit signal $x[n]$, but the
	received SI may include nonlinear and conjugate-dependent components. Let $\mathcal{P}=\{p_1,\ldots,p_{|\mathcal{P}|}\}$ denote the set of
	considered odd nonlinear orders, indexed in ascending order. For each
	$p\in\mathcal{P}$, define the direct nonlinear basis sequence
	\begin{equation}
		\phi_p^{(d)}[n]
		=
		x[n]|x[n]|^{p-1},
		\label{eq:direct_basis}
	\end{equation}
	and the conjugate nonlinear basis sequence
	\begin{equation}
		\phi_p^{(c)}[n]
		=
		x^*[n]|x[n]|^{p-1}.
		\label{eq:conjugate_basis}
	\end{equation}
	The direct terms model nonlinear distortion driven by the transmit waveform,
	whereas the conjugate terms model image-related widely-linear effects caused,
	for example, by I/Q imbalance.
	
	With memory, the effective behavioral SI model is
	\begin{equation}
		\begin{aligned}
			y_{\mathrm{SI}}[n]
			&\approx
			\sum_{p\in\mathcal{P}}
			\sum_{\ell=0}^{L_p^{(d)}-1}
			a_{p,\ell}^{(d)}
			\phi_p^{(d)}[n-\ell]
			\\
			&\quad+
			\sum_{p\in\mathcal{P}}
			\sum_{\ell=0}^{L_p^{(c)}-1}
			a_{p,\ell}^{(c)}
			\phi_p^{(c)}[n-\ell].
		\end{aligned}
		\label{eq:wl_mp_model}
	\end{equation}
	
	Here, $L_p^{(d)}$ and $L_p^{(c)}$ denote the memory lengths of the direct
	and conjugate branches, respectively, while $a_{p,\ell}^{(d)}$ and
	$a_{p,\ell}^{(c)}$ denote the corresponding SI coefficients.
	
	Define
	\begin{equation}
		\bm{\phi}_{p}^{(d)}[n]
		=
		\begin{bmatrix}
			\phi_p^{(d)}[n] &
			\phi_p^{(d)}[n-1] &
			\cdots &
			\phi_p^{(d)}[n-L_p^{(d)}+1]
		\end{bmatrix}^{T},
	\end{equation}
	and
	\begin{equation}
		\bm{\phi}_{p}^{(c)}[n]
		=
		\begin{bmatrix}
			\phi_p^{(c)}[n] &
			\phi_p^{(c)}[n-1] &
			\cdots &
			\phi_p^{(c)}[n-L_p^{(c)}+1]
		\end{bmatrix}^{T}.
	\end{equation}
	The full nonlinear feature vector is
	\begin{equation}
		\mathbf{g}[n]
		=
		\begin{bmatrix}
			\bm{\phi}_{p_1}^{(d)T}[n] &
			\bm{\phi}_{p_1}^{(c)T}[n] &
			\cdots &
			\bm{\phi}_{p_{|\mathcal{P}|}}^{(d)T}[n] &
			\bm{\phi}_{p_{|\mathcal{P}|}}^{(c)T}[n]
		\end{bmatrix}^{T}.
		\label{eq:feature_vector}
	\end{equation}
	Stacking the features over the training block gives
	\begin{equation}
		\mathbf{G}
		=
		\begin{bmatrix}
			\mathbf{g}^{T}[n_0]\\
			\mathbf{g}^{T}[n_0+1]\\
			\vdots\\
			\mathbf{g}^{T}[n_0+N-1]
		\end{bmatrix}
		\in\mathbb{C}^{N\times D},
		\label{eq:raw_dictionary}
	\end{equation}
	where
	\begin{equation}
		D
		=
		\sum_{p\in\mathcal{P}}
		\left(
		L_p^{(d)}+L_p^{(c)}
		\right)
	\end{equation}
	is the total number of direct and conjugate SI coefficients. This
	ordering also defines the branch-level groups used by the proposed
	group-sparse estimator: for each nonlinear order $p$, the direct-memory
coefficients and the conjugate-memory coefficients are treated as two separate
coefficient groups.
	
Let
$\mathbf{a}_{p}^{(d)}
=
[a_{p,0}^{(d)},\ldots,a_{p,L_p^{(d)}-1}^{(d)}]^T$
and
$\mathbf{a}_{p}^{(c)}
=
[a_{p,0}^{(c)},\ldots,a_{p,L_p^{(c)}-1}^{(c)}]^T$.
Then
\begin{equation}
	\mathbf{a}
	=
	\begin{bmatrix}
		\mathbf{a}_{p_1}^{(d)T} &
		\mathbf{a}_{p_1}^{(c)T} &
		\cdots &
		\mathbf{a}_{p_{|\mathcal{P}|}}^{(d)T} &
		\mathbf{a}_{p_{|\mathcal{P}|}}^{(c)T}
	\end{bmatrix}^{T}
	\in\mathbb{C}^{D}
\end{equation}
collects all unknown SI coefficients. The training model becomes
\begin{equation}
	\mathbf{y}
	=
	\mathbf{S}\mathbf{h}_u
	+
	\mathbf{G}\mathbf{a}
	+
	\mathbf{e},
	\label{eq:joint_training_model}
\end{equation}
where $\mathbf{e}\in\mathbb{C}^{N}$ is an effective disturbance vector that
collects the receiver-noise vector $\mathbf{w}$ and the residual SI modeling
error associated with the behavioral approximation in
\eqref{eq:wl_mp_model}. Thus, the objective is to estimate the two unknown
parameter vectors $\mathbf{h}_u$ and $\mathbf{a}$ in
\eqref{eq:joint_training_model}.
	
	\section{Pilot-Projected Estimation Model}
	\label{sec:pilot_projected_model}
	
	\subsection{Problem Motivation}
	
	We first focus on estimating $\mathbf{a}$ in
\eqref{eq:joint_training_model}. If $\mathbf{a}$ is estimated directly from
$\mathbf{y}$ using a large nonlinear dictionary, the dictionary may partially
explain the desired pilot. To prevent this, the contribution associated with
$\mathbf{h}_u$ is removed from the SI-estimation equation before the SI
coefficients are estimated.
	
	\begin{assumption}
		The pilot convolution matrix $\mathbf{S}$ has full column rank, i.e.,
		$\mathrm{rank}(\mathbf{S})=L_u$.
	\end{assumption}
	
	This standard condition is typically satisfied in practice and implies that
$\mathbf{S}^{H}\mathbf{S}$ is nonsingular.
	
\subsection{Pilot-Subspace Projection}

Define the orthogonal projector $\mathbf{P}_{\perp}$ onto the complement of the pilot subspace:
\begin{equation}
	\mathbf{P}_{\perp}
	=
	\mathbf{I}
	-
	\mathbf{S}(\mathbf{S}^{H}\mathbf{S})^{-1}\mathbf{S}^{H}.
	\label{eq:projection_matrix}
\end{equation}
Under the full-column-rank assumption on $\mathbf{S}$, $\mathbf{P}_{\perp}$ is
Hermitian and idempotent, and satisfies
$\mathbf{P}_{\perp}\mathbf{S}=\mathbf{0}$. Therefore, applying
$\mathbf{P}_{\perp}$ to \eqref{eq:joint_training_model} gives
\begin{equation}
	\mathbf{P}_{\perp}\mathbf{y}
	=
	\mathbf{P}_{\perp}\mathbf{G}\mathbf{a}
	+
	\mathbf{P}_{\perp}\mathbf{e}.
	\label{eq:projected_model_raw}
\end{equation}
Thus, the desired pilot component is removed before SI coefficient estimation.

The projected model can be represented either in the original $N$-dimensional
space or in a reduced-dimensional projected space. Let
$\mathbf{B}_{\perp}\in\mathbb{C}^{N\times N_p}$ contain an orthonormal basis for
the range of $\mathbf{P}_{\perp}$, i.e.,
\begin{equation}
	\mathbf{B}_{\perp}^{H}\mathbf{B}_{\perp}
	=
	\mathbf{I}_{N_p},
	\qquad
	\mathbf{B}_{\perp}\mathbf{B}_{\perp}^{H}
	=
	\mathbf{P}_{\perp},
\end{equation}
where
\begin{equation}
	N_p=\mathrm{rank}(\mathbf{P}_{\perp})=N-L_u.
	\label{eq:Np_def}
\end{equation}
The compact projected observation $\mathbf{y}_{p}$ and dictionary $\mathbf{A}_{0}$ are then
\begin{equation}
	\mathbf{y}_{p}
	=
	\mathbf{B}_{\perp}^{H}\mathbf{y},
	\qquad
	\mathbf{A}_{0}
	=
	\mathbf{B}_{\perp}^{H}\mathbf{G}.
\end{equation}
Hence,
\begin{equation}
	\mathbf{y}_{p}
	=
	\mathbf{A}_{0}\mathbf{a}
	+
	\mathbf{n}_{p},
	\qquad
	\mathbf{n}_{p}
	=
	\mathbf{B}_{\perp}^{H}\mathbf{e}.
	\label{eq:compact_projected_model}
\end{equation}

\begin{proposition}[Elimination of the uplink-channel nuisance parameter]
	\label{prop:nuisance_elimination}
	Let $\mathcal{R}(\mathbf{a})$ be any regularization term that depends on the SI
	coefficients but not on the uplink channel. Under the full-column-rank
	assumption on $\mathbf{S}$, the joint problem
\begin{equation}
\min_{\mathbf{h}_u,\mathbf{a}}
\left\|
\mathbf{y}
-
\mathbf{S}\mathbf{h}_u
-
\mathbf{G}\mathbf{a}
\right\|_2^2
+
\mathcal{R}(\mathbf{a})
\label{eq:joint_nuisance_problem}
\end{equation}
	has the same set of minimizers with respect to the SI-coefficient vector
$\mathbf{a}$ as the reduced problem
	\begin{equation}
		\min_{\mathbf{a}}
		\left\|
		\mathbf{B}_{\perp}^{H}\mathbf{y}
		-
		\mathbf{B}_{\perp}^{H}\mathbf{G}\mathbf{a}
		\right\|_2^2
		+
		\mathcal{R}(\mathbf{a}).
		\label{eq:reduced_nuisance_problem}
	\end{equation}
\end{proposition}

\begin{IEEEproof}
	For fixed $\mathbf{a}$, the least-squares minimizer of
\eqref{eq:joint_nuisance_problem} with respect to $\mathbf{h}_u$ is
\begin{equation}
\hat{\mathbf{h}}_u(\mathbf{a})
=
(\mathbf{S}^{H}\mathbf{S})^{-1}
\mathbf{S}^{H}
\left(
\mathbf{y}-\mathbf{G}\mathbf{a}
\right).
\end{equation}
	The corresponding residual is
	\begin{equation}
		\mathbf{y}
		-
		\mathbf{S}\hat{\mathbf{h}}_u(\mathbf{a})
		-
		\mathbf{G}\mathbf{a}
		=
		\mathbf{P}_{\perp}
		\left(
		\mathbf{y}-\mathbf{G}\mathbf{a}
		\right).
	\end{equation}
	Since
	$\|\mathbf{P}_{\perp}\mathbf{z}\|_2^2
	=
	\|\mathbf{B}_{\perp}^{H}\mathbf{z}\|_2^2$
	for any vector $\mathbf{z}$, substituting
	$\hat{\mathbf{h}}_u(\mathbf{a})$ into
\eqref{eq:joint_nuisance_problem}
	gives \eqref{eq:reduced_nuisance_problem}.
\end{IEEEproof}

Thus, pilot projection is the reduced form of
\eqref{eq:joint_nuisance_problem} after $\mathbf{h}_u$ has been profiled out
for SI-coefficient estimation. The subsequent estimators therefore operate only
on the SI component that remains observable outside the pilot subspace.

For the likelihood-based derivations, we use the circular Gaussian working model
$\mathbf{e}\sim\mathcal{CN}(\mathbf{0},\sigma_e^2\mathbf{I})$, where
$\sigma_e^2$ is an effective disturbance variance that may include thermal noise
and residual modeling error and therefore need not equal the receiver-noise
variance $\sigma^2$ in \eqref{eq:received_signal}.

\begin{remark}
	The compact representation avoids the singular covariance that would arise if
	one treated $\mathbf{P}_{\perp}\mathbf{e}$ as an $N$-dimensional Gaussian
	vector. In implementation, one may work directly with $\mathbf{P}_{\perp}$,
	since
	$\|\mathbf{P}_{\perp}\mathbf{z}\|_2^2
	=
	\|\mathbf{B}_{\perp}^{H}\mathbf{z}\|_2^2$.
\end{remark}

\subsection{Column Normalization}

The columns of $\mathbf{G}$ can have substantially different powers, especially
across nonlinear orders. Without normalization, regularization would depend on
the arbitrary scale of each dictionary column. Let $\mathbf{g}_j$ denote the
$j$th column of $\mathbf{G}$ and define
\begin{equation}
	c_j=\|\mathbf{g}_j\|_2.
\end{equation}
Let
\begin{equation}
	\mathbf{D}_{G}=\mathrm{diag}(c_1,\ldots,c_D).
\end{equation}
The normalized dictionary is
\begin{equation}
	\widetilde{\mathbf{G}}
	=
	\mathbf{G}\mathbf{D}_{G}^{-1}.
\end{equation}
The normalized coefficient vector $\bm{\beta}$ is defined through
\begin{equation}
\mathbf{G}\mathbf{a}
=
\widetilde{\mathbf{G}}\bm{\beta},
\qquad
\mathbf{a}
=
\mathbf{D}_{G}^{-1}\bm{\beta}.
\label{eq:normalized_coefficients}
\end{equation}
The compact projected normalized dictionary is
\begin{equation}
	\mathbf{A}
	=
	\mathbf{B}_{\perp}^{H}\widetilde{\mathbf{G}}
	=
	\mathbf{A}_{0}\mathbf{D}_{G}^{-1}.
	\label{eq:projected_normalized_dictionary}
\end{equation}
Therefore, the final projected normalized estimation model is
\begin{equation}
	\mathbf{y}_{p}
	=
	\mathbf{A}\bm{\beta}
	+
	\mathbf{n}_{p},
	\qquad
	\mathbf{n}_{p}\sim\mathcal{CN}(\mathbf{0},\sigma_e^2\mathbf{I}_{N_p}).
	\label{eq:projected_model}
\end{equation}
In the sequel, unless otherwise stated, the projected dictionary refers to the
normalized compact dictionary $\mathbf{A}$ in
\eqref{eq:projected_normalized_dictionary}.
	
\begin{remark}
Column normalization does not change the SI subspace spanned by the dictionary.
It only changes the coordinate system in which the coefficients are estimated.
Thus, by \eqref{eq:normalized_coefficients}, estimating $\bm{\beta}$ is
equivalent to estimating $\mathbf{a}$.
\end{remark}

\begin{remark}[Identifiable dictionary energy after projection]
	Let $\widetilde{\mathbf{g}}_j$ denote the $j$th column of the normalized SI dictionary matrix $\widetilde{\mathbf{G}}$,
	so that $\|\widetilde{\mathbf{g}}_j\|_2=1$, and define
	\begin{equation}
		\mathbf{P}_{S}
		=
		\mathbf{S}(\mathbf{S}^{H}\mathbf{S})^{-1}\mathbf{S}^{H}.
	\end{equation}
	The fraction of the $j$th SI basis waveform lying in the pilot subspace is
	\begin{equation}
		\rho_j
		=
		\|\mathbf{P}_{S}\widetilde{\mathbf{g}}_j\|_2^2,
		\qquad
		0\leq \rho_j\leq 1.
		\label{eq:pilot_alignment_factor}
	\end{equation}
	Since $\mathbf{P}_{\perp}=\mathbf{I}-\mathbf{P}_{S}$, the energy retained after
	projection is
	\begin{equation}
		\|\mathbf{B}_{\perp}^{H}\widetilde{\mathbf{g}}_j\|_2^2
		=
		\|\mathbf{P}_{\perp}\widetilde{\mathbf{g}}_j\|_2^2
		=
		1-\rho_j.
		\label{eq:retained_projected_energy}
	\end{equation}
	Thus, SI basis functions with $\rho_j\approx 1$ are weakly identifiable from the
	projected observation because most of their energy is indistinguishable from the
	protected pilot response over the training block. This loss is not specific to
	ridge or Bayesian estimation; it follows from enforcing invariance to the
	unknown uplink-channel coordinates.
\end{remark}
\begin{algorithm}[t]
		\caption{Common Pilot-Projected Nonlinear SI Estimation Front-End}
		\label{alg:common_frontend}
		\begin{algorithmic}[1]
			\State \textbf{Input:} received vector $\mathbf{y}$, pilot sequence $t[n]$, local waveform $x[n]$, order set $\mathcal{P}$, memory lengths $\{L_p^{(d)},L_p^{(c)}\}_{p\in\mathcal{P}}$
			\State Construct the pilot convolution matrix $\mathbf{S}$
			\State Generate $\phi_p^{(d)}[n]=x[n]|x[n]|^{p-1}$ for all $p\in\mathcal{P}$
			\State Generate $\phi_p^{(c)}[n]=x^*[n]|x[n]|^{p-1}$ for all $p\in\mathcal{P}$
			\State Construct the raw nonlinear dictionary $\mathbf{G}$
			\State Normalize the columns of $\mathbf{G}$ to obtain $\widetilde{\mathbf{G}}$
			\State Form $\mathbf{P}_{\perp}$ using \eqref{eq:projection_matrix}
			\State Construct, explicitly or implicitly, an orthonormal basis $\mathbf{B}_{\perp}$ for the range of $\mathbf{P}_{\perp}$
			\State Compute $\mathbf{y}_{p}=\mathbf{B}_{\perp}^{H}\mathbf{y}$
			\State Compute $\mathbf{A}=\mathbf{B}_{\perp}^{H}\widetilde{\mathbf{G}}$
			\State \textbf{Output:} $\mathbf{y}_{p}$, $\mathbf{A}$, and the scaling matrix $\mathbf{D}_{G}$
		\end{algorithmic}
	\end{algorithm}
	
	\section{BIC-Tuned Pilot-Projected Ridge Estimation}
	\label{sec:ridge_estimation}
	After pilot projection, the desired uplink-pilot component is removed from the
	training observation. However, the resulting SI estimation problem can still be
	ill conditioned. First, the widely-linear nonlinear dictionary may contain many
	correlated basis functions, especially when several nonlinear orders, memory
	taps, and direct and conjugate branches are included. Second, the projection
	onto the orthogonal complement of the pilot subspace may attenuate some SI basis
	directions, making the corresponding coefficients only weakly identifiable from
	the projected data. Consequently, an unregularized least-squares fit may amplify
	noise, residual modeling error, or weak singular directions of the projected
	dictionary. To obtain a stable full-dictionary estimator, this section introduces
	a pilot-projected ridge formulation and selects its regularization parameter.
	\subsection{Regularized Least-Squares Formulation}
	
	The pilot-projected ridge estimator solves
	\begin{equation}
		\hat{\bm{\beta}}_{\lambda}
		=
		\arg\min_{\bm{\beta}}
		\left\|
		\mathbf{y}_{p}
		-
		\mathbf{A}\bm{\beta}
		\right\|_2^2
		+
		\lambda\|\bm{\beta}\|_2^2,
		\label{eq:ridge_problem}
	\end{equation}
	where $\lambda > 0$ is the ridge regularization parameter. The solution is
	\begin{equation}
		\hat{\bm{\beta}}_{\lambda}
		=
		(\mathbf{A}^{H}\mathbf{A}+\lambda\mathbf{I})^{-1}
		\mathbf{A}^{H}\mathbf{y}_{p}.
		\label{eq:ridge_primal}
	\end{equation}

\subsection{Ridge-BIC Parameter Selection}
	
The ridge parameter controls the bias-variance tradeoff. If $\lambda$ is too
small, weak singular directions may amplify noise. If $\lambda$ is too large,
the SI estimate becomes overly biased. We therefore select $\lambda$ using a
ridge-adapted BIC. Unlike standard BIC in greedy model selection, which compares
discrete supports with different numbers of selected basis functions, the
proposed criterion keeps the full dictionary and replaces the hard parameter
count by the degrees of freedom of the ridge fit.
	
	Let
	\begin{equation}
		\Lambda=\{\lambda_1,\ldots,\lambda_{N_{\lambda}}\}
	\end{equation}
	denote a finite grid of candidate ridge parameters, where
	$N_{\lambda}=|\Lambda|$. For fixed $\lambda$, define
	the fitted projected SI component
	\begin{equation}
		\hat{\mathbf{y}}_{p,\lambda}
		=
		\mathbf{A}\hat{\bm{\beta}}_{\lambda}.
	\end{equation}
	Using \eqref{eq:ridge_primal},
	\begin{equation}
		\hat{\mathbf{y}}_{p,\lambda}
		=
		\mathbf{H}_{\lambda}\mathbf{y}_{p},
	\end{equation}
	where
	\begin{equation}
		\mathbf{H}_{\lambda}
		=
		\mathbf{A}
		(\mathbf{A}^{H}\mathbf{A}+\lambda\mathbf{I})^{-1}
		\mathbf{A}^{H}
		\label{eq:ridge_hat_matrix}
	\end{equation}
	is the ridge fitting matrix. In ordinary least squares, the trace of the fitting
	matrix equals the number of fitted parameters. Ridge does not perform hard
	parameter selection; it shrinks singular directions continuously. Therefore,
	$\mathrm{tr}(\mathbf{H}_{\lambda})$ is used as the effective number of degrees
	of freedom.
	
	\begin{proposition}
		Let $\{\sigma_i\}$ denote the singular values of $\mathbf{A}$. The effective
		degrees of freedom of the ridge estimator are
		\begin{equation}
			d_{\lambda}
			=
			\mathrm{tr}(\mathbf{H}_{\lambda})
			=
			\sum_i
			\frac{\sigma_i^2}{\sigma_i^2+\lambda}.
			\label{eq:ridge_df}
		\end{equation}
	\end{proposition}
	
	\begin{IEEEproof}
		Let the compact singular value decomposition of $\mathbf{A}$ be
		\begin{equation}
			\mathbf{A}=\mathbf{U}\bm{\Sigma}_{A}\mathbf{V}^{H},
		\end{equation}
		where the diagonal entries of $\bm{\Sigma}_{A}$ are the singular values
		$\{\sigma_i\}$. Substituting into \eqref{eq:ridge_hat_matrix} gives
		\begin{equation}
			\begin{aligned}
				\mathbf{H}_{\lambda}
				&=
				\mathbf{U}\bm{\Sigma}_{A}\mathbf{V}^{H}
				\left(
				\mathbf{V}\bm{\Sigma}_{A}^{H}\bm{\Sigma}_{A}\mathbf{V}^{H}
				+
				\lambda\mathbf{I}
				\right)^{-1}
				\mathbf{V}\bm{\Sigma}_{A}^{H}\mathbf{U}^{H} \\
				&=
				\mathbf{U}
				\bm{\Sigma}_{A}
				\left(
				\bm{\Sigma}_{A}^{H}\bm{\Sigma}_{A}
				+
				\lambda\mathbf{I}
				\right)^{-1}
				\bm{\Sigma}_{A}^{H}
				\mathbf{U}^{H}.
			\end{aligned}
		\end{equation}
		Hence, the eigenvalues of $\mathbf{H}_{\lambda}$ are
		$\sigma_i^2/(\sigma_i^2+\lambda)$, and taking the trace gives
		\eqref{eq:ridge_df}.
	\end{IEEEproof}
	
	Each term in \eqref{eq:ridge_df} lies between zero and one. If
	$\sigma_i^2\gg\lambda$, the corresponding singular direction contributes almost
	one degree of freedom. If $\sigma_i^2\ll\lambda$, that direction is strongly
	suppressed and contributes almost zero. Thus, $d_{\lambda}$ measures the
	effective model dimension after ridge shrinkage.
	
	For each $\lambda$, define the residual vector
	\begin{equation}
		\mathbf{r}_{\lambda}
		=
		\mathbf{y}_{p}
		-
		\mathbf{A}\hat{\bm{\beta}}_{\lambda},
	\end{equation}
	and the corresponding residual sum of squares (RSS)
	\begin{equation}
		\mathrm{RSS}_{\lambda}
		=
		\|\mathbf{r}_{\lambda}\|_2^2.
	\end{equation}
	Under the proper complex Gaussian working model in \eqref{eq:projected_model}, profiling over $\sigma_e^2$ gives, up to constants independent of $\lambda$, the complex-data negative log-likelihood term $N_p\log(\mathrm{RSS}_{\lambda}/N_p)$. We count both the sample size and the effective dimension in complex units: $N_p$ is the number of complex projected samples, and $d_{\lambda}$ is the complex effective dimension of the ridge smoother. With this convention, the ridge-BIC score is \begin{equation} \mathrm{BIC}^{\mathrm{ridge}}_{\lambda} = N_p \log \left( \frac{\mathrm{RSS}_{\lambda}}{N_p} \right) + d_{\lambda}\log N_p. \label{eq:ridge_bic} \end{equation}
	The first term rewards accurate fitting of the projected received signal, while
	the second term penalizes effective model complexity. This penalty is necessary
	because minimizing the residual alone would tend to favor weak regularization and
	overly flexible solutions, which may fit noise or model-mismatch artifacts rather
	than reproducible SI structure.
	
	The selected parameter is
	\begin{equation}
	\lambda^{\star}
	=
	\arg\min_{\lambda\in\Lambda}\mathrm{BIC}^{\mathrm{ridge}}_{\lambda}.
	\end{equation}
The final ridge estimate is
\begin{equation}
	\hat{\bm{\beta}}_{\mathrm{R}}
	=
	\hat{\bm{\beta}}_{\lambda^{\star}}.
\end{equation}
	
\subsection{Interpretation}

PP-RidgeBIC is a regularized full-dictionary estimator for the pilot-projected
SI model. It avoids hard selection of nonlinear basis functions and instead controls the
ill conditioning of the projected dictionary through a single ridge parameter
selected by $\mathrm{BIC}^{\mathrm{ridge}}_{\lambda}$. This makes it simple and numerically stable, but
also imposes the same shrinkage on all direct and conjugate components. It
therefore cannot express the fact that different nonlinear branches may have
different statistical relevance in the projected data. This limitation motivates
the Bayesian formulation in the next section, where the branch structure is
introduced through a probabilistic prior rather than through deterministic
preselection.
	
	\begin{algorithm}[t]
		\caption{BIC-Tuned Pilot-Projected Ridge Estimator (PP-RidgeBIC)}
		\label{alg:ridge}
		\begin{algorithmic}[1]
			\State \textbf{Input:} projected observation $\mathbf{y}_{p}$, projected dictionary $\mathbf{A}$, candidate set $\Lambda$
			\State Compute the singular values $\{\sigma_i\}$ of $\mathbf{A}$
			\State Initialize $\mathrm{BIC}^{\mathrm{ridge}}_{\min}=\infty$
			\For{each $\lambda\in\Lambda$}
		\State Compute $\hat{\bm{\beta}}_{\lambda}$ using \eqref{eq:ridge_primal}
			\State Compute
			\[
			\mathbf{r}_{\lambda}
			=
			\mathbf{y}_{p}-\mathbf{A}\hat{\bm{\beta}}_{\lambda}
			\]
			\State Compute $d_{\lambda}$ from \eqref{eq:ridge_df}
		\State Compute $\mathrm{BIC}^{\mathrm{ridge}}_{\lambda}$ from \eqref{eq:ridge_bic}
		\If{$\mathrm{BIC}^{\mathrm{ridge}}_{\lambda}<\mathrm{BIC}^{\mathrm{ridge}}_{\min}$}
		\State $\mathrm{BIC}^{\mathrm{ridge}}_{\min}\gets\mathrm{BIC}^{\mathrm{ridge}}_{\lambda}$
			\State $\lambda^{\star}\gets\lambda$
			\State $\hat{\bm{\beta}}_{\mathrm{R}}\gets\hat{\bm{\beta}}_{\lambda}$
			\EndIf
			\EndFor
			\State \textbf{Output:} $\hat{\bm{\beta}}_{\mathrm{R}}$ and $\lambda^{\star}$
		\end{algorithmic}
	\end{algorithm}
	
\section{Pilot-Projected Group-Sparse Bayesian Estimation}
\label{sec:gsbe_estimation}

\subsection{Motivation and Group Structure}
The normalized widely-linear memory-polynomial dictionary is naturally organized by nonlinear branches. For each nonlinear order $p$, the direct basis functions and the conjugate-dependent basis functions correspond to two separate branches, representing different possible contributions to the residual SI. The memory taps within each branch describe the memory response of the same nonlinear SI component and are therefore grouped to share a common regularization strength. Thus, each direct or conjugate branch constitutes exactly one group. Further splitting a branch into smaller groups would introduce
additional variance hyperparameters without representing distinct nonlinear
SI components and would reduce statistical sharing among the memory taps,
which is particularly undesirable under limited training. PP-GSBE uses this organization only at the prior level: the relevant branches are not assumed to be known in advance. Instead, the coefficients in each branch are assigned a common prior variance, and these variances are learned from the projected training data. 
	
	Let $N_{\mathrm{grp}}$ denote the number of groups, and let
	$\mathcal{G}_g\subseteq\{1,\ldots,D\}$ denote the index set of the coefficients
	belonging to group $g$. The size of group $g$ is
	\begin{equation}
		d_g = |\mathcal{G}_g|.
	\end{equation}
	The normalized coefficient vector is partitioned as
	\begin{equation}
		\bm{\beta}
		=
		\begin{bmatrix}
			\bm{\beta}_1^T &
			\bm{\beta}_2^T &
			\cdots &
			\bm{\beta}_{N_{\mathrm{grp}}}^T
		\end{bmatrix}^{T},
	\end{equation}
	where
	\begin{equation}
		\bm{\beta}_g=\bm{\beta}_{\mathcal{G}_g}
		\in\mathbb{C}^{d_g}
	\end{equation}
contains the coefficients of group $g$.
	
This partition specifies the level at which the Bayesian prior is tied across
coefficients. The term group-sparse therefore refers to Bayesian branch-wise
shrinkage: the relevance of each branch is inferred from the projected data
through its learned variance, without imposing an explicit hard support-selection
step.
	
	\subsection{Hierarchical Bayesian Model}
	
Under the projected observation model in \eqref{eq:projected_model}, the
likelihood is
	\begin{equation}
		p(\mathbf{y}_{p}|\bm{\beta},\sigma_e^2)
		=
		\frac{1}{(\pi\sigma_e^2)^{N_p}}
		\exp
		\left(
		-
		\frac{1}{\sigma_e^2}
		\|\mathbf{y}_{p}-\mathbf{A}\bm{\beta}\|_2^2
		\right).
		\label{eq:sbl_likelihood_density}
	\end{equation}
	
	Each group is assigned a zero-mean complex Gaussian prior:
	\begin{equation}
		\bm{\beta}_g
		\sim
		\mathcal{CN}
		\left(
		\mathbf{0},
		\gamma_g\mathbf{I}_{d_g}
		\right),
		\label{eq:group_prior}
	\end{equation}
	where $\gamma_g>0$ is the variance hyperparameter of group $g$. This use of
	variance hyperparameters for relevance learning follows the general sparse
	Bayesian learning theory~\cite{Tipping2001,Wipf2004}. Define
	\begin{equation}
		\bm{\Gamma}
		=
		\mathrm{diag}
		\left(
		\gamma_1\mathbf{I}_{d_1},
		\gamma_2\mathbf{I}_{d_2},
		\ldots,
		\gamma_{N_{\mathrm{grp}}}\mathbf{I}_{d_{N_{\mathrm{grp}}}}
		\right).
		\label{eq:Gamma}
	\end{equation}
	Let
$\bm{\gamma}
=
[\gamma_1,\ldots,\gamma_{N_{\mathrm{grp}}}]^T$.
Then
\begin{equation}
	\bm{\beta}
	\sim
	\mathcal{CN}(\mathbf{0},\bm{\Gamma}),
\end{equation}
with prior density
\begin{equation}
	p(\bm{\beta}|\bm{\gamma})
		=
		\frac{1}{\pi^D\det(\bm{\Gamma})}
		\exp
		\left(
		-
		\bm{\beta}^{H}\bm{\Gamma}^{-1}\bm{\beta}
		\right).
		\label{eq:beta_prior_density}
	\end{equation}

	\subsection{Posterior Distribution and Coefficient Estimate}
	
Let
\begin{equation}
	\bm{\theta}
	=
	\{\gamma_1,\ldots,\gamma_{N_{\mathrm{grp}}},\sigma_e^2\}
\end{equation}
denote the set of PP-GSBE hyperparameters. For fixed $\bm{\theta}$, the
posterior distribution of $\bm{\beta}$ is obtained from Bayes' rule:
	\begin{equation}
		p(\bm{\beta}|\mathbf{y}_{p},\bm{\theta})
		\propto
		p(\mathbf{y}_{p}|\bm{\beta},\sigma_e^2)
		p(\bm{\beta}|\bm{\gamma}).
	\end{equation}
	Substituting \eqref{eq:sbl_likelihood_density} and
	\eqref{eq:beta_prior_density}, and ignoring terms independent of $\bm{\beta}$,
	gives
	\begin{equation}
		\begin{aligned}
			\log p(\bm{\beta}|\mathbf{y}_{p},\bm{\theta})
			&=
			-
			\frac{1}{\sigma_e^2}
			\|\mathbf{y}_{p}-\mathbf{A}\bm{\beta}\|_2^2
			-
			\bm{\beta}^{H}\bm{\Gamma}^{-1}\bm{\beta}
			+
			\mathrm{const}.
		\end{aligned}
		\label{eq:posterior_log_start}
	\end{equation}
	Expanding the quadratic data-fitting term yields
	\begin{equation}
		\begin{aligned}
			\|\mathbf{y}_{p}-\mathbf{A}\bm{\beta}\|_2^2
			&=
			\mathbf{y}_{p}^{H}\mathbf{y}_{p}
			-
			\mathbf{y}_{p}^{H}\mathbf{A}\bm{\beta}
			-
			\bm{\beta}^{H}\mathbf{A}^{H}\mathbf{y}_{p}
			+
			\bm{\beta}^{H}\mathbf{A}^{H}\mathbf{A}\bm{\beta}.
		\end{aligned}
	\end{equation}
	The first term is independent of $\bm{\beta}$ and is absorbed into the
	constant. Hence,
	\begin{equation}
		\begin{aligned}
			\log p(\bm{\beta}|\mathbf{y}_{p},\bm{\theta})
			&=
			-
			\bm{\beta}^{H}
			\left(
			\frac{1}{\sigma_e^2}\mathbf{A}^{H}\mathbf{A}
			+
			\bm{\Gamma}^{-1}
			\right)
			\bm{\beta}
			\\
			&\quad
			+
			2\mathrm{Re}
			\left\{
			\bm{\beta}^{H}
			\frac{1}{\sigma_e^2}
			\mathbf{A}^{H}\mathbf{y}_{p}
			\right\}
			+
			\mathrm{const}.
		\end{aligned}
		\label{eq:posterior_log_quadratic}
	\end{equation}
	
	Define the posterior precision matrix
	\begin{equation}
		\bm{\Lambda}_{\beta}
		=
		\frac{1}{\sigma_e^2}\mathbf{A}^{H}\mathbf{A}
		+
		\bm{\Gamma}^{-1},
		\label{eq:sbl_precision}
	\end{equation}
	and the posterior linear term
	\begin{equation}
		\mathbf{b}_{\beta}
		=
		\frac{1}{\sigma_e^2}\mathbf{A}^{H}\mathbf{y}_{p}.
		\label{eq:sbl_b}
	\end{equation}
	Then the $\bm{\beta}$-dependent part of the posterior log-density is
	\begin{equation}
		-
		\bm{\beta}^{H}\bm{\Lambda}_{\beta}\bm{\beta}
		+
		2\mathrm{Re}\{\bm{\beta}^{H}\mathbf{b}_{\beta}\}.
	\end{equation}
	Let
	\begin{equation}
		\bm{\Sigma}
		=
		\bm{\Lambda}_{\beta}^{-1},
		\label{eq:sbl_cov}
	\end{equation}
	and
	\begin{equation}
		\bm{\mu}
		=
		\bm{\Lambda}_{\beta}^{-1}\mathbf{b}_{\beta}
		=
		\bm{\Sigma}
		\frac{1}{\sigma_e^2}
		\mathbf{A}^{H}\mathbf{y}_{p}.
		\label{eq:sbl_mean}
	\end{equation}
	Since $\bm{\Lambda}_{\beta}\bm{\mu}=\mathbf{b}_{\beta}$, completing the square
	gives
	\begin{equation}
		\begin{aligned}
			\bm{\beta}^{H}\bm{\Lambda}_{\beta}\bm{\beta}
			-
			2\mathrm{Re}\{\bm{\beta}^{H}\mathbf{b}_{\beta}\}
			&=
			(\bm{\beta}-\bm{\mu})^{H}
			\bm{\Lambda}_{\beta}
			(\bm{\beta}-\bm{\mu})
			-
			\bm{\mu}^{H}\bm{\Lambda}_{\beta}\bm{\mu}.
		\end{aligned}
	\end{equation}
	The last term is independent of $\bm{\beta}$ and is absorbed into the
	normalization constant. Therefore,
	\begin{equation}
	\bm{\beta}| \mathbf{y}_{p},\bm{\theta}
	\sim
		\mathcal{CN}(\bm{\mu},\bm{\Sigma}).
		\label{eq:sbl_posterior}
	\end{equation}
	
	Then, the coefficient estimate is chosen as the posterior mean:
	\begin{equation}
		\hat{\bm{\beta}}_{\mathrm{GSBE}}
		=
		\mathbb{E}_{\bm{\beta}|\mathbf{y}_{p},\bm{\theta}}
		[\bm{\beta}]
		=
		\bm{\mu}.
		\label{eq:gsbe_estimate}
	\end{equation}
	This is the minimum mean-square error estimate under squared-error loss:
	\begin{equation}
		\bm{\mu}
		=
		\arg\min_{\mathbf{z}}
		\mathbb{E}_{\bm{\beta}|\mathbf{y}_{p},\bm{\theta}}
		\left[
		\|\bm{\beta}-\mathbf{z}\|_2^2
		\right].
	\end{equation}
Since the posterior is Gaussian, the posterior mean also coincides with the
posterior mode. Thus, for fixed hyperparameters, the PP-GSBE estimate is also
the maximum a posteriori estimate.
	
	\subsection{Connection to Adaptive Group-Wise Ridge Regularization}
	
	The posterior mean can be written as
	\begin{equation}
		\bm{\mu}
		=
		\left(
		\mathbf{A}^{H}\mathbf{A}
		+
		\sigma_e^2\bm{\Gamma}^{-1}
		\right)^{-1}
		\mathbf{A}^{H}\mathbf{y}_{p}.
		\label{eq:sbl_as_ridge}
	\end{equation}
	This has the same form as ridge regression, except that the scalar ridge
	parameter is replaced by the matrix $\sigma_e^2\bm{\Gamma}^{-1}$. Since
	\begin{equation}
		\bm{\Gamma}^{-1}
		=
		\mathrm{diag}
		\left(
		\frac{1}{\gamma_1}\mathbf{I}_{d_1},
		\ldots,
		\frac{1}{\gamma_{N_{\mathrm{grp}}}}\mathbf{I}_{d_{N_{\mathrm{grp}}}}
		\right),
	\end{equation}
	the effective regularization strength of group $g$ is
	\begin{equation}
		\frac{\sigma_e^2}{\gamma_g}.
	\end{equation}
Thus, large $\gamma_g$ means weak shrinkage of group $g$, while small
$\gamma_g$ means strong shrinkage. PP-GSBE can therefore be interpreted as
ridge regression with a learned group-dependent regularization level.
	
	\subsection{Group Variance Update}
	
The group variances are updated through an expectation--maximization
(EM)-type procedure. At iteration $i$, let
	\begin{equation}
		\bm{\theta}^{(i)}
		=
		\{\gamma_1^{(i)},\ldots,\gamma_{N_{\mathrm{grp}}}^{(i)},
		(\sigma_e^2)^{(i)}\}
	\end{equation}
	denote the current hyperparameters. The posterior
	$p(\bm{\beta}|\mathbf{y}_{p},\bm{\theta}^{(i)})$ has mean
	$\bm{\mu}^{(i)}$ and covariance $\bm{\Sigma}^{(i)}$, computed using
	\eqref{eq:sbl_mean} and \eqref{eq:sbl_cov}, respectively. The next value
	$\gamma_g^{(i+1)}$ is then obtained by maximizing the expected log-prior of the
	coefficients in group $g$, with the expectation taken over this posterior
	distribution.
	
	For group $g$,
	\begin{equation}
		p(\bm{\beta}_g|\gamma_g)
		=
		\frac{1}{(\pi\gamma_g)^{d_g}}
		\exp
		\left(
		-
		\frac{1}{\gamma_g}
		\|\bm{\beta}_g\|_2^2
		\right),
	\end{equation}
	and hence
	\begin{equation}
		\log p(\bm{\beta}_g|\gamma_g)
		=
		-d_g\log\gamma_g
		-
		\frac{1}{\gamma_g}\|\bm{\beta}_g\|_2^2
		+
		\mathrm{const},
		\label{eq:group_prior_log}
	\end{equation}
	where the constant is independent of $\gamma_g$.
	
	Define the EM auxiliary function for $\gamma_g$ as
	\begin{equation}
		\mathcal{Q}_{\gamma_g}
		(\gamma_g|\bm{\theta}^{(i)})
		=
		\mathbb{E}_{\bm{\beta}|\mathbf{y}_{p},\bm{\theta}^{(i)}}
		\left[
		\log p(\bm{\beta}_g|\gamma_g)
		\right].
		\label{eq:Q_gamma_def}
	\end{equation}
	In this auxiliary function, the posterior distribution is fixed from the
	E-step, and $\gamma_g$ is the only variable optimized in the M-step.
	
	Substituting \eqref{eq:group_prior_log} into \eqref{eq:Q_gamma_def} gives
	\begin{equation}
		\mathcal{Q}_{\gamma_g}
		(\gamma_g|\bm{\theta}^{(i)})
		=
		-d_g\log\gamma_g
		-
		\frac{1}{\gamma_g}
		\mathbb{E}_{\bm{\beta}|\mathbf{y}_{p},\bm{\theta}^{(i)}}
		\left[
		\|\bm{\beta}_g\|_2^2
		\right]
		+
		\mathrm{const}.
		\label{eq:Q_gamma}
	\end{equation}
	Let
	\begin{equation}
		E_g^{(i)}
		=
		\mathbb{E}_{\bm{\beta}|\mathbf{y}_{p},\bm{\theta}^{(i)}}
		\left[
		\|\bm{\beta}_g\|_2^2
		\right].
	\end{equation}
	Then
	\begin{equation}
		\mathcal{Q}_{\gamma_g}
		(\gamma_g|\bm{\theta}^{(i)})
		=
		-d_g\log\gamma_g
		-
		\frac{E_g^{(i)}}{\gamma_g}
		+
		\mathrm{const}.
	\end{equation}
	Differentiating with respect to $\gamma_g$ gives
	\begin{equation}
		\frac{\partial \mathcal{Q}_{\gamma_g}}{\partial\gamma_g}
		=
		-\frac{d_g}{\gamma_g}
		+
		\frac{E_g^{(i)}}{\gamma_g^2}.
	\end{equation}
	Setting this derivative equal to zero yields
	\begin{equation}
		\gamma_g^{(i+1)}
		=
		\frac{E_g^{(i)}}{d_g}.
		\label{eq:gamma_Eg}
	\end{equation}
	
	It remains to compute $E_g^{(i)}$. Under the current posterior,
	\begin{equation}
		\bm{\beta}_g|\mathbf{y}_{p},\bm{\theta}^{(i)}
		\sim
		\mathcal{CN}(\bm{\mu}_g^{(i)},\bm{\Sigma}_g^{(i)}),
	\end{equation}
	where
	\begin{equation}
		\bm{\mu}_g^{(i)}
		=
		\bm{\mu}_{\mathcal{G}_g}^{(i)},
		\qquad
		\bm{\Sigma}_g^{(i)}
		=
		\bm{\Sigma}_{\mathcal{G}_g,\mathcal{G}_g}^{(i)}.
	\end{equation}
	Therefore,
	\begin{equation}
		E_g^{(i)}
		=
		\|\bm{\mu}_g^{(i)}\|_2^2
		+
		\mathrm{tr}(\bm{\Sigma}_g^{(i)}).
	\end{equation}
	Substituting into \eqref{eq:gamma_Eg} gives
	\begin{equation}
		\gamma_g^{(i+1)}
		=
		\frac{
			\|\bm{\mu}_g^{(i)}\|_2^2
			+
			\mathrm{tr}(\bm{\Sigma}_g^{(i)})
		}
		{d_g}.
		\label{eq:gamma_update}
	\end{equation}
	
Thus, the updated $\gamma_g$ is the posterior average power per coefficient in
group $g$.
	
	\subsection{Noise Variance Update}
	
	The effective projected noise variance is updated similarly. At iteration $i$,
	define the EM auxiliary function for $\sigma_e^2$ as
	\begin{equation}
		\mathcal{Q}_{\sigma}
		(\sigma_e^2|\bm{\theta}^{(i)})
		=
		\mathbb{E}_{\bm{\beta}|\mathbf{y}_{p},\bm{\theta}^{(i)}}
		\left[
		\log p(\mathbf{y}_{p}|\bm{\beta},\sigma_e^2)
		\right].
		\label{eq:Q_sigma_def}
	\end{equation}
	This function is the expected log-likelihood of the projected observation,
	where the expectation is taken with respect to the current posterior
	$p(\bm{\beta}|\mathbf{y}_{p},\bm{\theta}^{(i)})$, and $\sigma_e^2$ is treated as
	the variable to be optimized.
	
	From \eqref{eq:sbl_likelihood_density}, the log-likelihood term depending on
	$\sigma_e^2$ is
	\begin{equation}
		\log p(\mathbf{y}_{p}|\bm{\beta},\sigma_e^2)
		=
		-N_p\log\sigma_e^2
		-
		\frac{1}{\sigma_e^2}
		\|\mathbf{y}_{p}-\mathbf{A}\bm{\beta}\|_2^2
		+
		\mathrm{const}.
	\end{equation}
	Therefore,
	\begin{equation}
		\mathcal{Q}_{\sigma}
		(\sigma_e^2|\bm{\theta}^{(i)})
		=
		-N_p\log\sigma_e^2
		-
		\frac{1}{\sigma_e^2}
		E_{\mathrm{res}}^{(i)}
		+
		\mathrm{const},
	\end{equation}
	where
	\begin{equation}
		E_{\mathrm{res}}^{(i)}
		=
		\mathbb{E}_{\bm{\beta}|\mathbf{y}_{p},\bm{\theta}^{(i)}}
		\left[
		\|\mathbf{y}_{p}-\mathbf{A}\bm{\beta}\|_2^2
		\right].
	\end{equation}
	Differentiating with respect to $\sigma_e^2$ gives
	\begin{equation}
		\frac{\partial \mathcal{Q}_{\sigma}}{\partial\sigma_e^2}
		=
		-\frac{N_p}{\sigma_e^2}
		+
		\frac{E_{\mathrm{res}}^{(i)}}{(\sigma_e^2)^2}.
	\end{equation}
	Setting this derivative equal to zero yields
	\begin{equation}
		(\sigma_e^2)^{(i+1)}
		=
		\frac{E_{\mathrm{res}}^{(i)}}{N_p}.
		\label{eq:sigma_Eres}
	\end{equation}
	
It remains to express $E_{\mathrm{res}}^{(i)}$. Since
	\begin{equation}
		\bm{\beta}
		=
		\bm{\mu}^{(i)}
		+
		\left(
		\bm{\beta}-\bm{\mu}^{(i)}
		\right),
	\end{equation}
	we have
	\begin{equation}
		\mathbf{y}_{p}-\mathbf{A}\bm{\beta}
		=
		\mathbf{y}_{p}-\mathbf{A}\bm{\mu}^{(i)}
		-
		\mathbf{A}
		\left(
		\bm{\beta}-\bm{\mu}^{(i)}
		\right).
	\end{equation}
	Because
	\begin{equation}
		\mathbb{E}_{\bm{\beta}|\mathbf{y}_{p},\bm{\theta}^{(i)}}
		\left[
		\bm{\beta}-\bm{\mu}^{(i)}
		\right]
		=
		\mathbf{0},
	\end{equation}
	the cross term vanishes. Therefore,
	\begin{equation}
		\begin{aligned}
			E_{\mathrm{res}}^{(i)}
			&=
			\|\mathbf{y}_{p}-\mathbf{A}\bm{\mu}^{(i)}\|_2^2
			\\
			&\quad+
			\mathbb{E}_{\bm{\beta}|\mathbf{y}_{p},\bm{\theta}^{(i)}}
			\left[
			\left\|
			\mathbf{A}
			\left(
			\bm{\beta}-\bm{\mu}^{(i)}
			\right)
			\right\|_2^2
			\right].
		\end{aligned}
	\end{equation}
	Using
	\begin{equation}
		\mathbb{E}_{\bm{\beta}|\mathbf{y}_{p},\bm{\theta}^{(i)}}
		\left[
		\left(
		\bm{\beta}-\bm{\mu}^{(i)}
		\right)
		\left(
		\bm{\beta}-\bm{\mu}^{(i)}
		\right)^H
		\right]
		=
		\bm{\Sigma}^{(i)},
	\end{equation}
	we obtain
	\begin{equation}
		E_{\mathrm{res}}^{(i)}
		=
		\|\mathbf{y}_{p}-\mathbf{A}\bm{\mu}^{(i)}\|_2^2
		+
		\mathrm{tr}
		\left(
		\mathbf{A}\bm{\Sigma}^{(i)}\mathbf{A}^{H}
		\right).
	\end{equation}
	Thus, the update is
	\begin{equation}
		(\sigma_e^2)^{(i+1)}
		=
		\frac{
			\|\mathbf{y}_{p}-\mathbf{A}\bm{\mu}^{(i)}\|_2^2
			+
			\mathrm{tr}
			\left(
			\mathbf{A}\bm{\Sigma}^{(i)}\mathbf{A}^{H}
			\right)
		}
		{N_p}.
		\label{eq:noise_update}
	\end{equation}
	The first term is the residual energy obtained with the posterior mean, and the
	second term accounts for uncertainty in the coefficient estimate.
	
	\begin{remark}
	In practice, the hyperparameters are bounded to avoid numerical singularities.
	For example, one may enforce
	$\gamma_{\min}\le\gamma_g\le\gamma_{\max}$ and
	$\sigma_e^2\ge\sigma_{\min}^2$, where $\gamma_{\min}>0$ and
	$\sigma_{\min}^2>0$ are small positive lower bounds, and
	$\gamma_{\max}>\gamma_{\min}$ is an upper bound on the group variances.
	These bounds do not change the interpretation of the method; they only
	prevent ill-conditioned matrix inversions when a group is assigned
	very strong shrinkage.
	\end{remark}
	
\subsection{Algorithm and Interpretation}

The resulting algorithm alternates between posterior inference and
hyperparameter learning. Given the current values of $\{\gamma_g\}$ and
$\sigma_e^2$, the posterior mean and covariance are computed from
\eqref{eq:sbl_mean} and \eqref{eq:sbl_cov}, respectively. Then the group variances are updated
according to \eqref{eq:gamma_update}, and the effective noise variance is updated
according to \eqref{eq:noise_update}. This procedure is repeated until the
hyperparameters stabilize.

The branch-wise shrinkage mechanism is driven by the feedback between posterior
group energy and group variance. If group $g$ is weakly supported by the
projected received signal, its posterior energy becomes small, which reduces
$\gamma_g$ in the next iteration. Since the effective shrinkage applied to that
group increases as $\sigma_e^2/\gamma_g$ grows, a smaller variance estimate
attenuates the corresponding branch more strongly. Conversely, groups with
larger posterior energy receive larger variance estimates and experience weaker
shrinkage. Thus, PP-GSBE performs soft, data-driven branch-wise Bayesian
regularization without greedy hard selection. Compared with PP-RidgeBIC, this
branch-adaptive behavior increases the computational cost, but provides greater
flexibility in suppressing weak nonlinear orders and image branches under short
training or model mismatch.

\begin{algorithm}[t]
	\caption{Pilot-Projected Group-Sparse Bayesian Estimator (PP-GSBE)}
	\label{alg:gsbe}
		\begin{algorithmic}[1]
			\State \textbf{Input:} projected observation $\mathbf{y}_{p}$, projected dictionary $\mathbf{A}$, group partition $\{\mathcal{G}_g\}_{g=1}^{N_{\mathrm{grp}}}$, maximum iteration number $I_{\max}$, tolerance $\epsilon$
			\State Initialize $\gamma_g^{(0)}=1$ for all groups
			\State Initialize $(\sigma_e^2)^{(0)}$ from the projected received power
			\For{$i=0,\ldots,I_{\max}-1$}
			\State Form $\bm{\Gamma}^{(i)}$ from $\{\gamma_g^{(i)}\}$
			\State Compute
			\[
			\bm{\Lambda}_{\beta}^{(i)}
			=
			\frac{1}{(\sigma_e^2)^{(i)}}\mathbf{A}^{H}\mathbf{A}
			+
			\left(\bm{\Gamma}^{(i)}\right)^{-1}
			\]
			\State Compute
			\[
			\bm{\Sigma}^{(i)}
			=
			\left(\bm{\Lambda}_{\beta}^{(i)}\right)^{-1}
			\]
			\State Compute
			\[
			\bm{\mu}^{(i)}
			=
			\bm{\Sigma}^{(i)}
			\frac{1}{(\sigma_e^2)^{(i)}}
			\mathbf{A}^{H}\mathbf{y}_{p}
			\]
			\For{$g=1,\ldots,N_{\mathrm{grp}}$}
			\State Extract
			\[
			\bm{\mu}_g^{(i)}
			=
			\bm{\mu}_{\mathcal{G}_g}^{(i)},
			\qquad
			\bm{\Sigma}_g^{(i)}
			=
			\bm{\Sigma}_{\mathcal{G}_g,\mathcal{G}_g}^{(i)}
			\]
			\State Update $\gamma_g^{(i+1)}$ using \eqref{eq:gamma_update}
			\EndFor
			\State Update $(\sigma_e^2)^{(i+1)}$ using \eqref{eq:noise_update}
			\State Define
			\[
			\bm{\gamma}^{(i)}
			=
			[\gamma_1^{(i)},\ldots,\gamma_{N_{\mathrm{grp}}}^{(i)}]^T
			\]
			\State Define
			\[
			\bm{\gamma}^{(i+1)}
			=
			[\gamma_1^{(i+1)},\ldots,\gamma_{N_{\mathrm{grp}}}^{(i+1)}]^T
			\]
			\State Compute
			\[
			\Delta_{\gamma}
			=
			\frac{
				\|\bm{\gamma}^{(i+1)}-\bm{\gamma}^{(i)}\|_2
			}
			{
				\|\bm{\gamma}^{(i)}\|_2
			}
			\]
			\State Compute
			\[
			\Delta_{\sigma}
			=
			\frac{
				|(\sigma_e^2)^{(i+1)}-(\sigma_e^2)^{(i)}|
			}
			{
				|(\sigma_e^2)^{(i)}|
			}
			\]
			\If{$\max\{\Delta_{\gamma},\Delta_{\sigma}\}<\epsilon$}
			\State \textbf{break}
			\EndIf
			\EndFor
			\State Recompute $\bm{\mu}$ using the final hyperparameters
			\State Set $\hat{\bm{\beta}}_{\mathrm{GSBE}}=\bm{\mu}$
			\State \textbf{Output:} $\hat{\bm{\beta}}_{\mathrm{GSBE}}$, learned $\{\gamma_g\}$, learned $\sigma_e^2$
		\end{algorithmic}
	\end{algorithm}

	\section{Computational Complexity}
	
	In this section, we characterize the computational and memory complexities of the common pilot-projected front-end and the two proposed estimators, PP-RidgeBIC and PP-GSBE. 
    
We use the dimensions $N_p$ and $D$ defined above and the maximum iteration
number $I_{\max}$ specified in Algorithm~\ref{alg:gsbe}.
	
The
	raw nonlinear dictionary $\mathbf{G}$ is constructed from the known local
	transmit samples with cost $\mathcal{O}(ND)$. Column normalization also costs
	$\mathcal{O}(ND)$. The pilot-subspace projection can be applied without forming
	a dense $N\times N$ projection matrix. Using
	\begin{equation}
		\mathbf{P}_{\perp}\mathbf{z}
		=
		\mathbf{z}
		-
		\mathbf{S}(\mathbf{S}^{H}\mathbf{S})^{-1}\mathbf{S}^{H}\mathbf{z},
	\end{equation}
	forming $\mathbf{S}^{H}\mathbf{S}$ costs $\mathcal{O}(NL_u^2)$, and applying the
	projection to all dictionary columns costs
	$\mathcal{O}(NL_uD+L_u^2D)$. Since $L_u$ is usually much smaller than $N$, this
	front-end cost is dominated by dictionary construction, normalization, and
	projection.
	
	For PP-RidgeBIC, the singular values of
	$\mathbf{A}\in\mathbb{C}^{N_p\times D}$ are required to compute the effective
	degrees of freedom in \eqref{eq:ridge_df}. With
	\begin{equation}
		m=\min\{N_p,D\},
	\end{equation}
	a compact singular-value decomposition of $\mathbf{A}$ costs
	\begin{equation}
		\mathcal{O}(N_pD\,m).
	\end{equation}
	For each candidate $\lambda\in\Lambda$, the ridge estimate is
	computed by solving the $D\times D$ linear system in
	\eqref{eq:ridge_primal}, with a per-candidate cost of
	$\mathcal{O}(D^3)$. Computing the fitted signal and residual
	costs $\mathcal{O}(N_pD)$. Therefore, the total complexity
	of PP-RidgeBIC is
	\begin{equation}
	    \mathcal{O}\left(
	N_pD\,m + N_pD^2
	+ N_{\lambda}(D^3+N_pD)
	\right).
	\end{equation}
	
	For PP-GSBE, the Gram matrix $\mathbf{A}^{H}\mathbf{A}$ and the matched vector
	$\mathbf{A}^{H}\mathbf{y}_{p}$ can be precomputed with costs
	$\mathcal{O}(N_pD^2)$ and $\mathcal{O}(N_pD)$, respectively. At each iteration,
	forming the posterior precision matrix in Algorithm~\ref{alg:gsbe} costs
	$\mathcal{O}(D^2)$ after this precomputation. Computing the posterior covariance
	requires solving a $D\times D$ linear system, whose dominant cost is
	$\mathcal{O}(D^3)$. Computing the posterior mean costs $\mathcal{O}(D^2)$. The
	group-variance updates require the posterior energies of the predefined groups,
	and the noise-variance update in \eqref{eq:noise_update} requires the residual
	term and the trace term. With $\mathbf{A}^{H}\mathbf{A}$ precomputed, these
	updates cost $\mathcal{O}(D^2+N_pD)$ per iteration. Hence, the worst-case
	complexity of PP-GSBE is
	\begin{equation}
		\mathcal{O}
		\left(
		N_pD^2
		+
		I_{\max}
		\left(
		D^3+D^2+N_pD
		\right)
		\right).
	\end{equation}
	If the stopping criterion in Algorithm~\ref{alg:gsbe} is satisfied before
	$I_{\max}$ iterations, then $I_{\max}$ is replaced by the actual number of
	executed iterations. The final posterior-mean recomputation has the same order
	as one additional iteration and is absorbed in the above expression.
	
The memory requirement of PP-RidgeBIC is dominated by storing the projected
dictionary and the factors of the $D\times D$ primal ridge system, giving
\begin{equation}
	\mathcal{O}(N_pD+D^2).
\end{equation}
The memory requirement of PP-GSBE is dominated by storing the projected
dictionary, the Gram matrix, and the posterior covariance, giving
\begin{equation}
	\mathcal{O}(N_pD+D^2).
\end{equation}
	
	Therefore, PP-RidgeBIC has lower computational complexity because it performs a
	finite regularized grid search with one global shrinkage parameter. PP-GSBE has
	higher complexity because it repeatedly updates the posterior mean, posterior
	covariance, group variances, and effective noise variance. 

	\section{Uplink Channel Estimation and Data-Phase SI Reconstruction}
	\label{sec:channel_estimation_data_reconstruction}
	
	For either estimator, let $\hat{\bm{\beta}}$ denote either
	$\hat{\bm{\beta}}_{\mathrm{R}}$ or
	$\hat{\bm{\beta}}_{\mathrm{GSBE}}$. Through
	\eqref{eq:normalized_coefficients}, $\hat{\bm{\beta}}$ represents the estimated
	SI coefficients in normalized coordinates. We next estimate $\mathbf{h}_u$ in
	\eqref{eq:joint_training_model} from the SI-cleaned training observation. The
	reconstructed training-phase SI is
	\begin{equation}
		\hat{\mathbf{y}}_{\mathrm{SI,tr}}
		=
		\widetilde{\mathbf{G}}\hat{\bm{\beta}}.
		\label{eq:training_si_reconstruction}
	\end{equation}
	The cleaned training observation is
	\begin{equation}
		\tilde{\mathbf{y}}_{\mathrm{tr}}
		=
		\mathbf{y}
		-
		\hat{\mathbf{y}}_{\mathrm{SI,tr}}.
		\label{eq:cleaned_training_signal}
	\end{equation}
	Since the pilot convolution matrix $\mathbf{S}$ is known and has full column
	rank, the uplink channel is estimated by least squares as
	\begin{equation}
	\hat{\mathbf{h}}_{u}
=
\arg\min_{\mathbf{h}_u}
\left\|
\tilde{\mathbf{y}}_{\mathrm{tr}}
-
\mathbf{S}\mathbf{h}_u
\right\|_2^2
		=
		(\mathbf{S}^{H}\mathbf{S})^{-1}\mathbf{S}^{H}
		\tilde{\mathbf{y}}_{\mathrm{tr}}.
		\label{eq:uplink_channel_ls}
	\end{equation}
	Let the training-phase SI reconstruction error be
	\begin{equation}
	\bm{\delta}_{\mathrm{tr}}
=
\mathbf{y}_{\mathrm{SI}}
-
\hat{\mathbf{y}}_{\mathrm{SI,tr}}.
	\end{equation}
	Then
	\begin{equation}
		\tilde{\mathbf{y}}_{\mathrm{tr}}
		=
		\mathbf{S}\mathbf{h}_u
		+
		\bm{\delta}_{\mathrm{tr}}
		+
		\mathbf{w},
	\end{equation}
	and the resulting least-squares channel-estimation error is
	\begin{equation}
		\hat{\mathbf{h}}_{u}-\mathbf{h}_{u}
		=
		(\mathbf{S}^{H}\mathbf{S})^{-1}
		\mathbf{S}^{H}
		\left(
		\bm{\delta}_{\mathrm{tr}}+\mathbf{w}
		\right).
		\label{eq:channel_error_from_si_residual}
	\end{equation}
	Thus, residual SI reconstruction error enters the channel estimate through the
	pilot matched-filtering operation.
	
	During data transmission, let $\mathbf{y}_{\mathrm{data}}$,
$\mathbf{y}_{u,\mathrm{data}}$, and $\mathbf{y}_{\mathrm{SI,data}}$ denote the
received, desired-uplink, and residual-SI data-phase vectors, respectively,
and let $\mathbf{G}_{\mathrm{data}}$ denote the raw data-phase dictionary
built from the local transmit samples using the same nonlinear orders and
memory lengths as in training. The normalized data-phase
	dictionary is
	\begin{equation}
		\widetilde{\mathbf{G}}_{\mathrm{data}}
		=
		\mathbf{G}_{\mathrm{data}}\mathbf{D}_{G}^{-1},
	\end{equation}
	where $\mathbf{D}_{G}$ is the column-normalization matrix obtained from the
	training dictionary. The reconstructed data-phase SI is
	\begin{equation}
		\hat{\mathbf{y}}_{\mathrm{SI,data}}
		=
		\widetilde{\mathbf{G}}_{\mathrm{data}}\hat{\bm{\beta}}.
		\label{eq:data_si_reconstruction}
	\end{equation}
	The digitally cancelled data-phase received vector is
	\begin{equation}
		\tilde{\mathbf{y}}_{\mathrm{data}}
		=
		\mathbf{y}_{\mathrm{data}}
		-
		\hat{\mathbf{y}}_{\mathrm{SI,data}}.
		\label{eq:data_cancelled}
	\end{equation}
	The channel estimate $\hat{\mathbf{h}}_{u}$ in \eqref{eq:uplink_channel_ls} is
	then used for data detection over $\tilde{\mathbf{y}}_{\mathrm{data}}$.

\section{Numerical Evaluation}
\label{sec:numerical_results}

This section evaluates the proposed pilot-projected nonlinear SI cancellation
framework in a single-antenna in-band full-duplex link. The simulations follow
the signal model in Sections~\ref{sec:system_model}
and~\ref{sec:pilot_projected_model}. The true residual SI is generated by a
nonlinear RF-chain data-generating model, whereas all cancellers use the finite
widely-linear memory-polynomial dictionary in Section~\ref{sec:wl_dictionary}.
Thus, the results evaluate both pilot-contaminated SI estimation and robustness
to nonlinear model mismatch.

\subsection{Simulation Setup and Metrics}

The desired uplink uses quadrature phase-shift keying (QPSK) modulation over
an $L_u=3$ tap frequency-selective channel. The local transmit waveform is generated as a colored complex Gaussian
sequence using a first-order autoregressive filter with coefficient $0.95$,
followed by power normalization. The true residual SI is generated by a compact
RF-chain model with a memoryless saturating power amplifier (PA), I/Q image
leakage, and an $L_{\rm SI}=8$ tap coupling channel: $x_{\rm PA}[n]=x[n]/(1+
\alpha_{\rm PA}|x[n]|^2)$, $x_{\rm IQ}[n]=x_{\rm PA}[n]+
\epsilon_{\rm IQ}x_{\rm PA}^*[n]$, and $y_{\rm SI}[n]=
\sum_{\ell=0}^{L_{\rm SI}-1}h_{\rm SI}[\ell]x_{\rm IQ}[n-\ell]$,
where $h_{\rm SI}[\ell]$ denotes the $\ell$th tap of the SI coupling channel.
Unless otherwise stated, the PA nonlinearity strength and I/Q image-leakage
amplitude are set to $\alpha_{\rm PA}=0.10$ and $\epsilon_{\rm IQ}=0.03$,
respectively.

The estimation dictionary uses odd nonlinear orders
$\mathcal{P}=\{1,3,5,7\}$, direct-branch memory lengths $\{12,10,8,6\}$, and
conjugate-branch memory lengths $\{10,8,6,4\}$. Hence, the dictionary contains
$D=64$ basis functions, grouped into direct and conjugate nonlinear branches.
The dictionary columns are normalized over the training block, and the same
normalization factors are used for data-block SI reconstruction.

The desired-uplink signal-to-noise ratio is $30$ dB and the pre-digital
residual INR is $35$ dB,
corresponding to a pre-digital desired-signal-to-residual-SI ratio of $-5$ dB.
The training length is varied as
\[
N\in\{128,192,256,384,512,768,2000\},
\]
where the larger values are included mainly to illustrate the convergence trend
rather than typical uplink-pilot lengths. In conventional SI-only training,
longer training blocks are easier to obtain because no simultaneous uplink pilot
is required. The data block contains $N_{\rm data}=20000$ QPSK symbols. The
training-length results are averaged over $50$ Monte Carlo realizations. In the
robustness sweeps, $N=512$ and $20$ Monte Carlo realizations are used.

Since the SI model is used for residual SI reconstruction and cancellation
rather than identification of physical RF parameters, all signal-domain
metrics are evaluated over the independent data block. Let $y_{\rm SI,data}[n]$, $\hat{y}_{\rm SI,data}[n]$,
$y_{u,\rm data}[n]$, and $\tilde{y}_{\rm data}[n]$ denote the $n$th entries
of $\mathbf{y}_{\mathrm{SI,data}}$,
$\hat{\mathbf{y}}_{\mathrm{SI,data}}$,
$\mathbf{y}_{u,\mathrm{data}}$, and
$\tilde{\mathbf{y}}_{\mathrm{data}}$, respectively. Let
\[
P_{\rm SI}
=
\frac{1}{N_{\rm data}}
\sum_{n=1}^{N_{\rm data}} |y_{\rm SI,data}[n]|^2
\]
be the data-block SI power before digital cancellation, and let
\[
P_{\rm res}
=
\frac{1}{N_{\rm data}}
\sum_{n=1}^{N_{\rm data}}
|y_{\rm SI,data}[n]-\hat{y}_{\rm SI,data}[n]|^2
\]
be the residual SI reconstruction-error power after digital cancellation. The
digital self-interference-cancellation (SIC) gain and residual INR are defined as
\[
{\rm SIC}
=
10\log_{10}
\left(
\frac{P_{\rm SI}}{P_{\rm res}}
\right),
\qquad
{\rm INR}_{\rm res}
=
10\log_{10}
\left(
\frac{P_{\rm res}}{\sigma^2}
\right),
\]
where $\sigma^2$ is as defined in \eqref{eq:received_signal}. The pre-digital
residual INR used in the robustness sweep is
\[
{\rm INR}_{\rm pre}
=
10\log_{10}
\left(
\frac{P_{\rm SI}}{\sigma^2}
\right).
\]
The dashed noise-floor line in the residual-INR figures corresponds to
${\rm INR}_{\rm res}=0$ dB.

Using the least-squares channel estimate in \eqref{eq:uplink_channel_ls}, the
channel-estimation NMSE is
\[
{\rm NMSE}_{h}
=
\frac{\|\hat{\mathbf{h}}_u-\mathbf{h}_u\|_2^2}
{\|\mathbf{h}_u\|_2^2}.
\]
The BER is computed after data-phase SI cancellation, uplink channel estimation,
and QPSK sequence detection. The received-signal EVM is defined as
\[
{\rm EVM}
=
\sqrt{
	\frac{
		\frac{1}{N_{\rm data}}
		\sum_{n=1}^{N_{\rm data}}
		|\tilde{y}_{\rm data}[n]-y_{u,\rm data}[n]|^2
	}{
		\frac{1}{N_{\rm data}}
		\sum_{n=1}^{N_{\rm data}}
		|y_{u,\rm data}[n]|^2
	}
},
\]
and is plotted in dB as $20\log_{10}({\rm EVM})$.

\subsection{Benchmark Schemes}

We compare the two proposed estimators with three nonlinear SIC baselines. The
baselines are chosen to isolate three effects: the role of eliminating the
uplink-channel nuisance component, the effect of greedy hard selection after
projection, and the performance of orthonormalized adaptive polynomial
cancellation under the same projected training model. We do not include a
separate joint estimator that regularizes the SI coefficients while leaving the
uplink channel unregularized, because Proposition~\ref{prop:nuisance_elimination}
shows that such a formulation reduces to the pilot-projected problem after the
uplink-channel nuisance parameter is eliminated.

\textbf{NoProj-Ridge} uses the same full widely-linear nonlinear dictionary as
PP-RidgeBIC, but fits the coefficients directly from the unprojected received
training signal. Its ridge parameter is selected using the same ridge-BIC rule
as PP-RidgeBIC. This baseline follows the widely-linear nonlinear SIC modeling
framework used for full-duplex radios~\cite{Korpi2014} and tests whether
ridge regularization without pilot projection can prevent desired-pilot leakage.

\textbf{PP-DOMP}, the pilot-projected doubly orthogonal matching pursuit
(DOMP) baseline, first applies the pilot projection and then performs greedy
basis selection on the projected dictionary. Specifically, DOMP is used to
select nonlinear basis functions, and standard BIC is used as the
stopping/model-selection criterion. This baseline
is motivated by greedy nonlinear basis-selection methods for digital
SIC~\cite{Shahghasi2025} and by BIC-based model selection~\cite{Schwarz1978}.
It tests whether hard nonlinear model selection is sufficient after the
desired-pilot subspace has been removed.

\textbf{PP-AOP}, the pilot-projected adaptive orthonormal-polynomial baseline,
first applies the pilot projection, then orthonormalizes the projected nonlinear
dictionary and updates the coefficients using a normalized LMS-type recursion. This baseline represents adaptive
orthonormalized nonlinear cancellation approaches related to augmented nonlinear
adaptive filtering and orthonormalized polynomial modeling~\cite{Li2018,Lee2025}.
It tests the benefit of orthonormalized adaptive polynomial cancellation under
the same pilot-projected observation model.

The proposed schemes are \textbf{PP-RidgeBIC} and \textbf{PP-GSBE}.
PP-RidgeBIC uses the pilot-projected full dictionary and selects its ridge
parameter by minimizing $\mathrm{BIC}^{\mathrm{ridge}}_{\lambda}$ over the grid
$\{10^{-10},10^{-9.5},\ldots,10^2\}$. PP-GSBE uses the same projected
dictionary with branch-wise Bayesian shrinkage and does not require BIC-based
parameter selection. PP-GSBE uses at most $150$ iterations, tolerance $10^{-5}$,
and group-variance bounds $10^{-10}$ and $10^8$.

\begin{figure}[t]
	\centering
	\includegraphics[width=.8\linewidth]{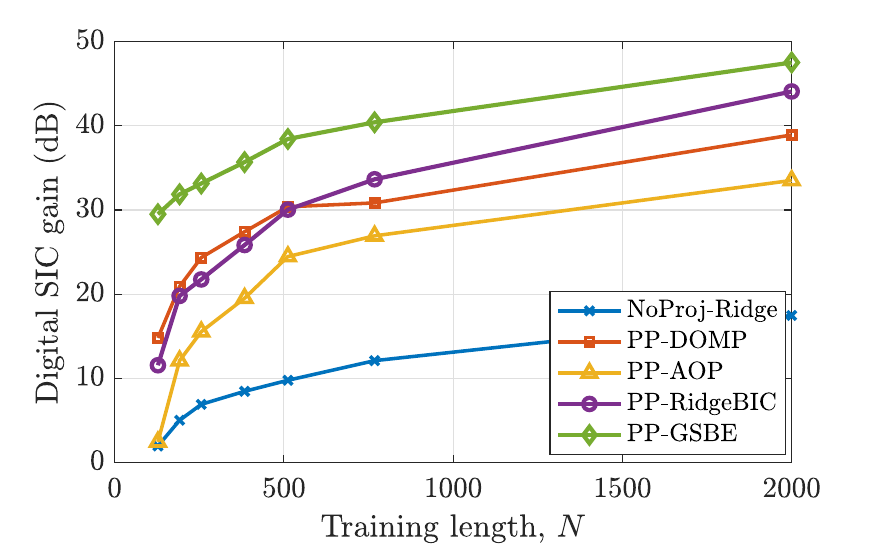}
	\caption{Digital SI cancellation gain versus training length.}
	\label{fig:sic_vs_training}
\end{figure}

\begin{figure}[t]
	\centering
	\includegraphics[width=.8\linewidth]{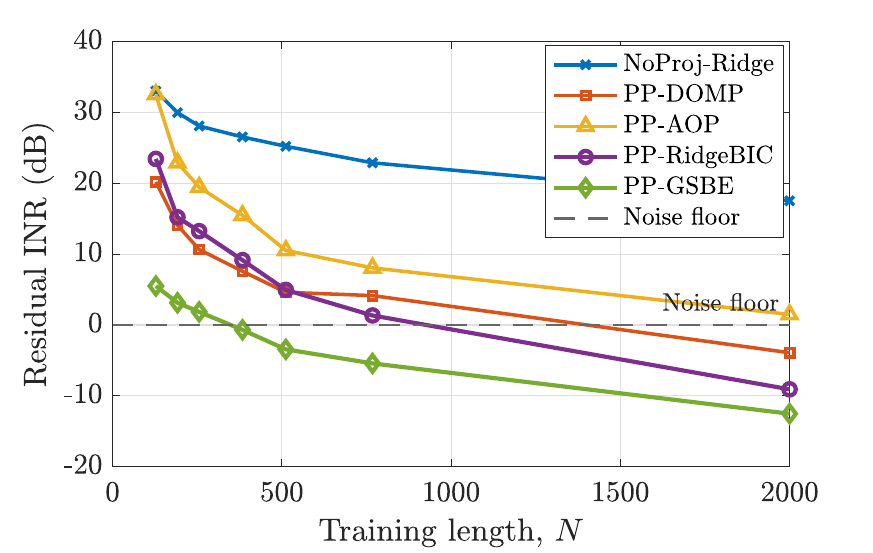}
	\caption{Residual INR after digital SI cancellation versus training length.
		The dashed line denotes the noise floor.}
	\label{fig:inr_vs_training}
\end{figure}

\subsection{Training-Length Results}

Fig.~\ref{fig:sic_vs_training} shows the digital SIC gain versus training
length. NoProj-Ridge provides only limited cancellation, confirming that
regularization without pilot projection cannot reliably separate the desired
pilot from the nonlinear SI dictionary. All projected schemes provide substantial
gains, demonstrating the importance of removing the pilot subspace before SI
estimation. PP-GSBE achieves the highest SIC gain over the full training range
because its branch-wise Bayesian shrinkage attenuates weak nonlinear branches
while retaining useful direct and conjugate components. At $N=2000$, PP-GSBE
reaches about $47$--$48$ dB, followed by PP-RidgeBIC, PP-DOMP, PP-AOP, and
NoProj-Ridge.

Fig.~\ref{fig:inr_vs_training} shows the corresponding residual INR. PP-GSBE is
the first scheme to push the residual SI below the noise floor. Around $N=512$,
it already achieves negative residual INR, whereas PP-DOMP and PP-RidgeBIC
remain several dB above the noise floor. At $N=2000$, PP-GSBE reaches
approximately $-13$ dB residual INR, showing that the proposed Bayesian
shrinkage improves data-block SI reconstruction under model mismatch.

\begin{figure}[t]
	\centering
	\includegraphics[width=.8\linewidth]{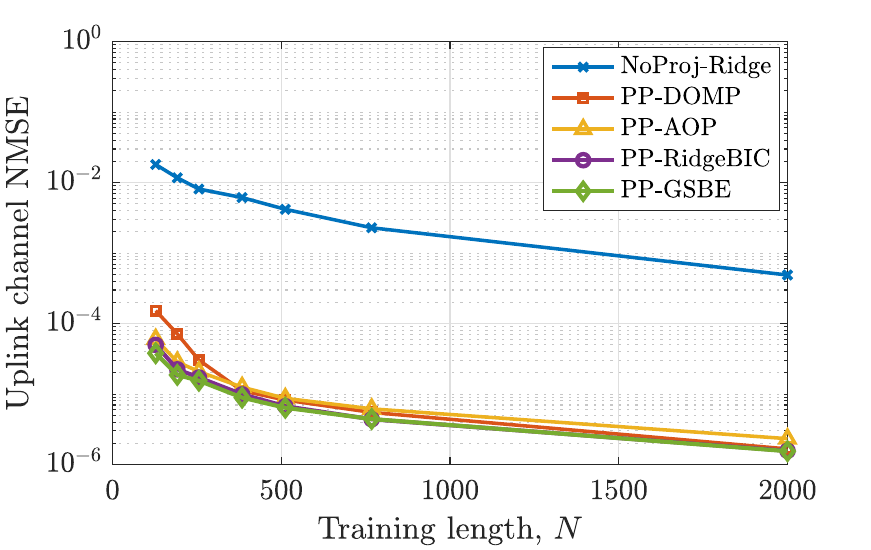}
	\caption{Uplink channel-estimation NMSE after digital SI cancellation.}
	\label{fig:nmse_vs_training}
\end{figure}

\begin{figure}[t]
	\centering
	\includegraphics[width=.8\linewidth]{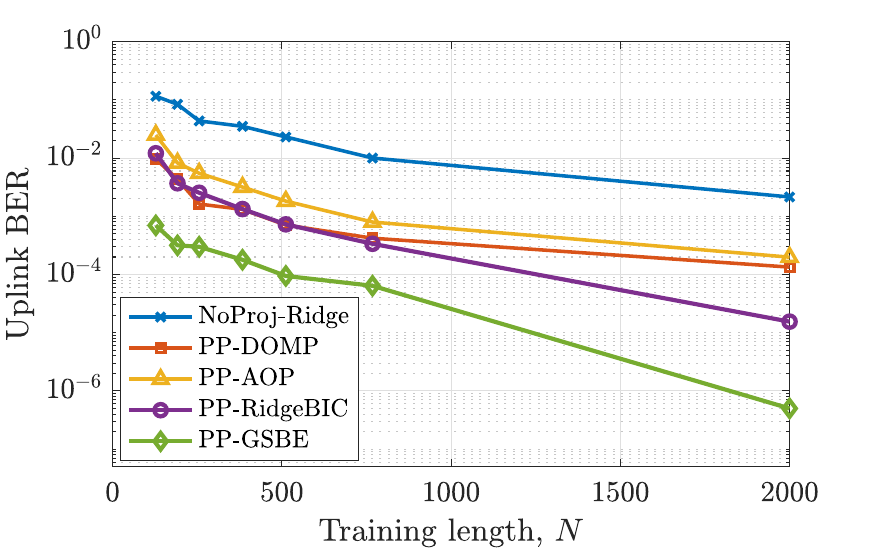}
	\caption{Uplink BER after digital SI cancellation, channel estimation, and
		QPSK sequence detection.}
	\label{fig:ber_vs_training}
\end{figure}

Fig.~\ref{fig:nmse_vs_training} reports the uplink channel-estimation NMSE after
SI cancellation. NoProj-Ridge has much larger NMSE because its SI estimate can
absorb part of the desired pilot. The projected schemes significantly reduce the NMSE and exhibit similar
performance because they all remove the same pilot subspace before SI estimation; as \eqref{eq:channel_error_from_si_residual} shows, only the
component of the residual SI error within this subspace affects the channel
estimate. PP-GSBE gives the best NMSE, particularly for short and moderate
training lengths.

Fig.~\ref{fig:ber_vs_training} shows the uplink BER after SI cancellation,
channel estimation, and data detection. The BER trend is consistent with the
residual-INR and NMSE results. NoProj-Ridge remains BER-limited, while the
projected methods provide large reliability gains. PP-DOMP improves when enough
training samples are available, but hard support selection is less robust for
short training blocks. PP-GSBE achieves the lowest BER over the considered range
and approaches the finite-simulation error floor at large $N$.

\begin{figure}[t]
	\centering
	\includegraphics[width=.8\linewidth]{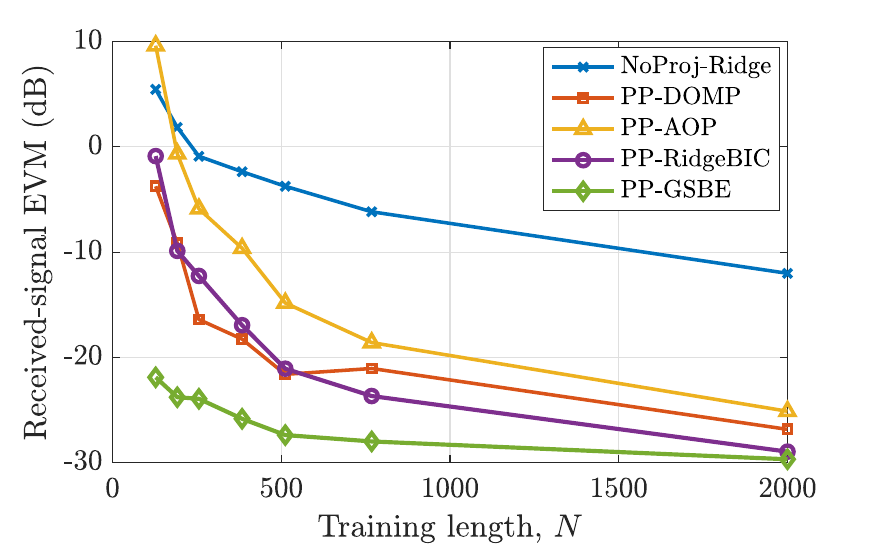}
	\caption{Received-signal EVM after digital SI cancellation.}
	\label{fig:evm_vs_training}
\end{figure}

Fig.~\ref{fig:evm_vs_training} shows the received-signal EVM after digital SI
cancellation. Since the true SI waveform is not exactly represented by the finite
estimation dictionary, EVM is a useful signal-domain measure of the cleaned
received waveform. PP-GSBE achieves the lowest EVM for all training lengths and
approaches about $-30$ dB at $N=2000$, confirming that the proposed estimator
improves the actual received-signal quality, not only the fitted training
residual.

\begin{figure}[t]
	\centering
	\includegraphics[width=.8\linewidth]{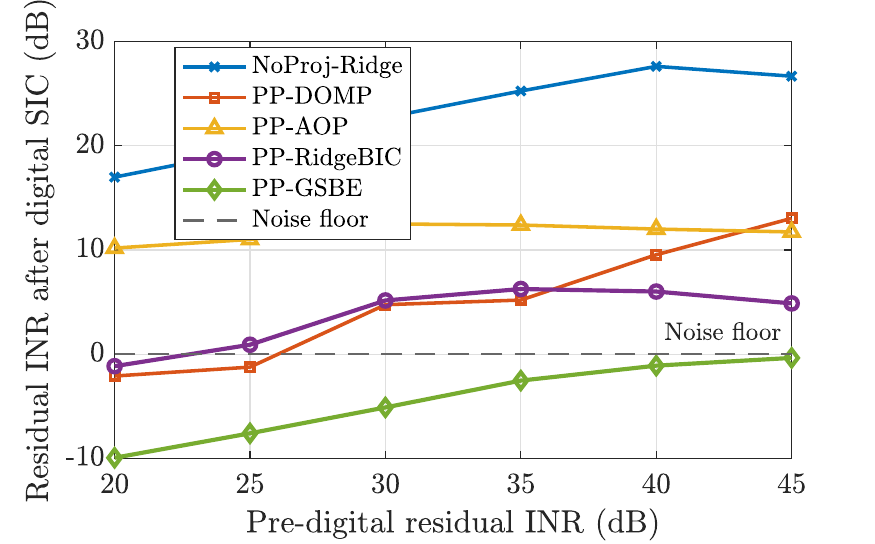}
	\caption{Residual INR after digital SI cancellation versus pre-digital
		residual INR. The training length is fixed to $N=512$.}
	\label{fig:inr_vs_preinr}
\end{figure}

\subsection{Robustness Results}

Fig.~\ref{fig:inr_vs_preinr} shows the residual INR versus the pre-digital
residual INR. This experiment evaluates robustness to the amount of SI left by
the analog/RF cancellation stage. NoProj-Ridge deteriorates rapidly as the
pre-digital INR increases. PP-AOP remains above the noise floor, while PP-DOMP
and PP-RidgeBIC are effective at low and moderate INR but degrade at high INR.
PP-GSBE is the most robust method, keeping the residual SI below or close to the
noise floor over almost the entire sweep.

\begin{figure}[t]
	\centering
	\includegraphics[width=.8\linewidth]{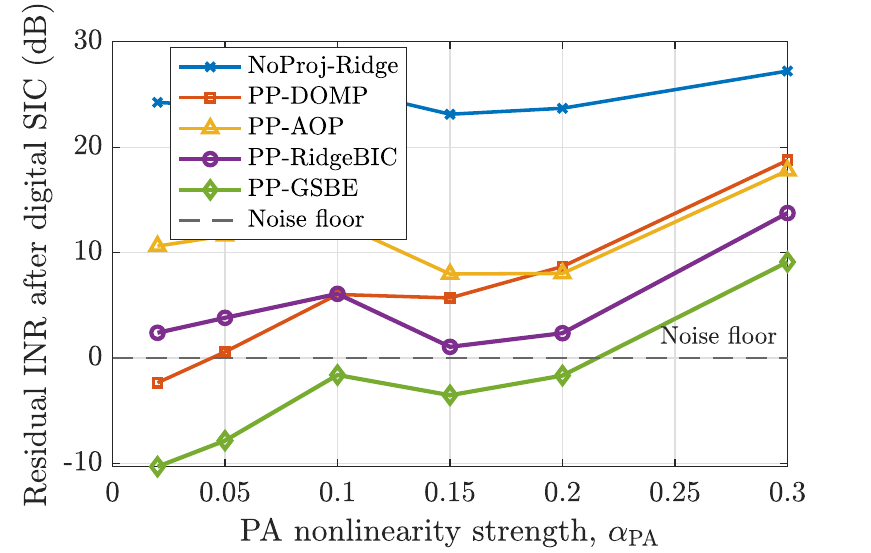}
	\caption{Residual INR after digital SI cancellation versus PA nonlinearity
		strength $\alpha_{\rm PA}$. The training length is fixed to $N=512$.}
	\label{fig:inr_vs_pa}
\end{figure}

Fig.~\ref{fig:inr_vs_pa} studies the effect of PA nonlinearity. Increasing
$\alpha_{\rm PA}$ changes the nonlinear SI waveform and increases the mismatch
between the true RF-chain response and the finite estimation dictionary.
PP-GSBE gives the lowest residual INR for all tested values and degrades more
gracefully under strong nonlinearity. The curves are not strictly monotonic
because changing $\alpha_{\rm PA}$ also changes how the normalized SI waveform
projects onto the finite dictionary.

\begin{figure}[t]
	\centering
	\includegraphics[width=.8\linewidth]{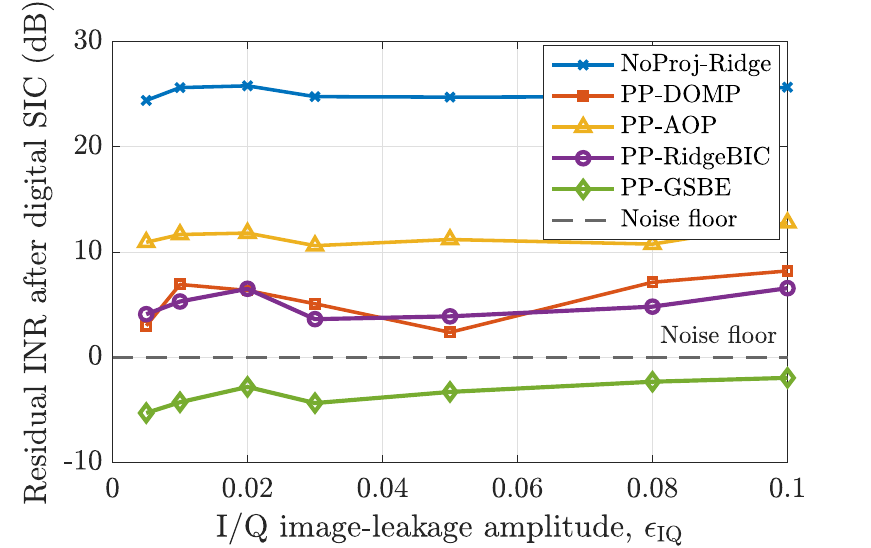}
	\caption{Residual INR after digital SI cancellation versus I/Q image-leakage
		amplitude $\epsilon_{\rm IQ}$. The training length is fixed to $N=512$.}
	\label{fig:inr_vs_iq}
\end{figure}

Fig.~\ref{fig:inr_vs_iq} shows the residual INR versus I/Q image-leakage
amplitude. This experiment evaluates whether the cancellers can exploit the
conjugate branches of the widely-linear dictionary when image-like SI components
become stronger. PP-GSBE remains below the noise floor over the full sweep,
indicating that its branch-wise shrinkage effectively retains useful
conjugate branches while attenuating unnecessary components.

Overall, the results show that pilot projection is essential when the desired
uplink pilot is present during nonlinear SI training. Full-dictionary
regularization improves stability, but uniform shrinkage is less flexible than
group-sparse Bayesian shrinkage. Hard selection is useful when sufficient
training data are available, but is less robust under short training and model
mismatch. The proposed PP-GSBE consistently provides the best performance across
SIC gain, residual INR, channel NMSE, BER, EVM, and robustness to residual SI
power, PA nonlinearity, and I/Q image leakage.

\section{Conclusion}

This paper developed a pilot-projected estimation framework for nonlinear
digital SI cancellation and uplink-channel estimation in in-band full-duplex
radios. The uplink-channel contribution was eliminated from SI-coefficient
estimation through projection onto the orthogonal complement of the pilot
subspace, after which the channel was estimated from the SI-cleaned training
observation. A normalized widely-linear memory-polynomial
dictionary was used to represent direct nonlinear and conjugate image components
of the residual SI, and its branch structure was exploited to define coefficient
groups. Two estimators that retain all candidate dictionary columns were developed: the BIC-tuned pilot-projected ridge estimator PP-RidgeBIC and the
pilot-projected group-sparse Bayesian estimator PP-GSBE. PP-RidgeBIC provides a stable, lower-complexity uniform-regularization
solution, whereas PP-GSBE provides a
more computationally demanding but more adaptive group-sparse solution with
branch-wise shrinkage learning. The framework clarifies how pilot projection,
dictionary normalization, regularization, and group-sparse Bayesian shrinkage
improve nonlinear digital SI cancellation under desired-pilot contamination,
ill-conditioned nonlinear dictionaries, and nonlinear model mismatch.

\end{document}